\documentclass[11pt,letterpaper,dvipsnames]{article} 
\usepackage[margin=1in]{geometry}
\usepackage{graphicx}
\usepackage[sort&compress,numbers]{natbib}
\usepackage[hyphens]{url}

\usepackage{caption} 
\usepackage{algorithm}
\usepackage{algorithmic}

\usepackage{booktabs} 
\usepackage{subcaption}

\usepackage{amsmath,amsfonts,amsthm}
\usepackage{tikz,units}
\usetikzlibrary{arrows,automata}
\usepackage{calc}
\usetikzlibrary{arrows.meta}
\usepackage{gensymb}
\usepackage{comment}
\usepackage{color}
\usepackage{MnSymbol}
\usepackage{multicol,multirow}
\usepackage{microtype}

\usepackage{tikz-3dplot}
\tikzset{%
	pics/ellip/.style args={#1}{code={%
			\fill[rotate=60,#1] (0,0) ellipse (0.9pt and 1.1pt);
	}},
	pics/ellip/.default=CarnationPink, 
    > = stealth, 
	auto,
	node distance = 3cm, 
	every edge/.append style = {thick}, 
    shorten >= -1pt,
    shorten <= -1pt,
}
\tikzstyle{axis}=[thin]
\tikzstyle{state}=
[circle,draw=black,thick,fill=black,minimum size=2mm,inner sep=0mm]

\title{Information Design for Congestion Games with Unknown Demand}
\author {
    Svenja M. Griesbach\thanks{Institute of Mathematics, Technische Universität Berlin, Germany, \texttt{$\{$griesbach,klimm$\}$@math.tu-berlin.de}} \and
    Martin Hoefer\thanks{Institute for Computer Science, Goethe University Frankfurt, Germany, \texttt{$\{$mhoefer,koglin$\}$@em.uni-frankfurt.de}} \and
    Max Klimm\footnotemark[1] \and
    Tim Koglin\footnotemark[2]
}
\date{}

\usepackage{cleveref}
\newtheorem{theorem}{Theorem}
\newtheorem{lemma}{Lemma}
\newtheorem{proposition}{Proposition}
\newtheorem{corollary}[lemma]{Corollary}
\theoremstyle{remark}
\newtheorem{example}{Example}

\crefname{figure}{Fig.}{Figs.}
\crefname{section}{Sec.}{Secs.}
\crefname{corollary}{Corollary}{Corollaries}

\newcommand{\R}{\mathbb{R}}
\renewcommand{\P}{\mathbb{P}}

\newcommand{\calP}{\mathcal{P}}

\newcommand{\calX}{\mathcal{X}}
\newcommand{\thetasum}{\theta \in \Theta}

\newcommand{\cost}{\ensuremath{C}}
\newcommand{\prior}{\mu^*}
\newcommand{\support}{A}

\newcommand{\NN}{\mathbb{N}}

\newcommand{\classP}{\textsf{P}}
\newcommand{\classNP}{\textsf{NP}}

\renewcommand{\phi}{\varphi}
\renewcommand{\epsilon}{\varepsilon}
\renewcommand{\setminus}{\, \backslash \,}

\newcommand{\D}{\displaystyle}

\allowdisplaybreaks

\begin{document}

\maketitle

\begin{abstract}
We study a novel approach to information design in the standard traffic model of network congestion games. It captures the natural condition that the \emph{demand} is unknown to the users of the network. A principal (e.g., a mobility service) commits to a signaling strategy, observes the realized demand and sends a (public) signal to agents (i.e., users of the network). Based on the induced belief about the demand, the users then form an equilibrium. We consider the algorithmic goal of the principal: Compute a signaling scheme that minimizes the expected total cost of the induced equilibrium. We concentrate on single-commodity networks and affine cost functions, for which we obtain the following results.
First, we devise a fully polynomial-time approximation scheme (FPTAS) for the case that the demand can only take two values. It relies on several structural properties of the cost of the induced equilibrium as a function of the updated belief about the distribution of demands. We show that this function is piecewise linear for any number of demands, and monotonic for two demands. 
Second, we give a complete characterization of the graph structures for which it is optimal to fully reveal the information about the realized demand. This signaling scheme turns out to be optimal for all cost functions and probability distributions over demands if and only if the graph is series-parallel. Third, we propose an algorithm that computes the optimal signaling scheme for any number of demands whose time complexity is polynomial in the number of supports that occur in a Wardrop equilibrium for some demand. Finally, we conduct a computational study that tests this algorithm on real-world instances. 
\end{abstract}

\section{Introduction}
\label{sec:introduction}
Traffic and congestion are key factors contributing to climate change and air pollution. On the other hand, personal and commercial traffic are fundamental for economic development and the modern way of life. This makes sound traffic planning and improvement an indispensable prerequisite for urban areas around the globe. A popular and successful model for traffic planning are non-atomic congestion games. The road network is represented by a graph $G=(V,E)$ where each edge~$e$ has a cost function $c_e$ that models the time needed to traverse the edge and depends on the total flow on that edge. In the single-commodity setting, a continuum of players with travel demand $d>0$ strives to route from a designated source vertex $s \in V$ (e.g., a residential living area) to a designated destination vertex $t\in V$ (e.g., a city center). Each infinitesimally small player aims to minimize their private cost by choosing a least-cost path from $s$ to $t$. A so-called \emph{Wardrop equilibrium} is reached when no player has the incentive to deviate from their chosen path since all other paths have either the same or even higher cost. 
It is a well-known fact that a Wardrop equilibrium does not necessarily minimize the overall travel time, and there is a substantial literature that quantifies the loss in efficiency due to selfish behavior \cite{Dubey86,RoughgardenT02,Roughgarden03,RoughgardenT04,CorreaSM04,CorreaSM08,DumraufG06}.

In order to achieve better equilibria, interventions through network design~\cite{Marcotte86,BhaskarLS14,GairingHK17,Roughgarden06} or mechanism design techniques such as tolls \cite{HarksKKM15,FleischerJM04,HoeferOS08,LarssonP99,HearnR98,BergendorffHR97} 
have been studied extensively. These approaches, however, usually come with a high cost, e.g., for building or remodeling road segments or for setting up a toll collection system for highways. 
This paper, therefore, focuses on improving the emerging equilibrium by \emph{information design}.

A significant source of uncertainty in traffic networks concerns the \emph{demand}, i.e., information about the total amount of traffic. Total traffic is highly fluctuating, even during a single day. From a game-theoretic perspective, this implies that the total volume of players in the routing game is not fixed and not common knowledge. Such \emph{games with population uncertainty} were first studied in a systematic way by \citet{Myerson1998}, who considered games with atomic players and multiple player types. \citet{CominettiSSM22} draw a connection between the Poisson games of \citeauthor{Myerson1998} and atomic congestion games where players participate independently at random. Here we adapt the approach to non-atomic games. The total volume of players in the game is drawn from a probability distribution known to all players.
Each player observes whether they participate in the game or not, e.g., whether to drive to work in the morning (i.e., has type \emph{active}) or not (type \emph{inactive}, e.g., due to illness or car malfunction). When a player is inactive and does not participate, they receive a private cost of $0$. Otherwise, an active player then fixes a strategy, i.e., an $s$-$t$-path in the network, and receives as private cost the cost of the chosen path. This leads to a Bayesian game in the sense of \citet{Harsanyi1967}. 

Due to technical reasons, in our analysis we treat a slightly less intuitive model variant. In this version, players decide on a route \emph{before} they know whether they are active or not. An inactive player simply discards the route choice made earlier. The two model variants lead to equivalent outcomes in terms of equilibria and cost (see \Cref{app:alternative}). The latter variant allows to avoid a uniform scaling factor in all computations. Hence, despite being less intuitive, we use it throughout the paper for simplicity. For illustration, we start by discussing a simple example.

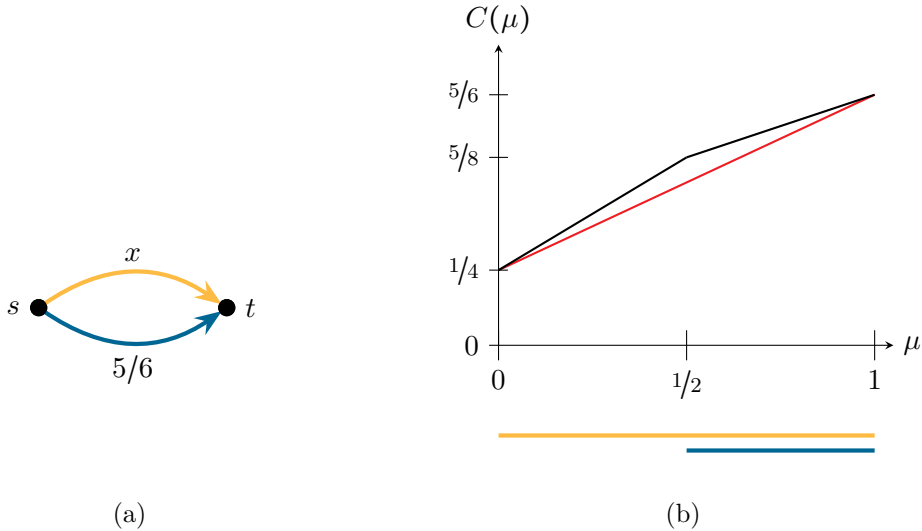
\begin{figure}[t]
\begin{center}
\begin{subfigure}[b]{0.4\textwidth}
\begin{center}
\begin{tikzpicture}[xscale=10/4,yscale=3]
\tikzstyle{node}=[circle, fill=gray, inner sep=0pt, minimum size=3pt];
\begin{scope}
\node[state,label=right:{$t$}] (t1) at (0.5,0.4) {};
\node[state,label=left:{$s$}] (s1) at (-0.5,0.4) {};

\draw[-{Stealth},ultra thick,Dandelion] (s1) to[bend left] node[above] {\color{black}$x$} (t1);
\draw[-{Stealth},ultra thick,MidnightBlue] (s1) to[bend right] node[below] {\color{black}$5/6$} (t1);

\node[state] (t1) at (0.5,0.4) {};
\node[state] (s1) at (-0.5,0.4) {};
\node[] (heightcheat) at (-0.5,-0.2) {};
\end{scope}
\end{tikzpicture}
\end{center}
\caption{}
\end{subfigure}
\hspace{0.5cm}
\begin{subfigure}[b]{0.4\textwidth}
\begin{center}
\begin{tikzpicture}[xscale=5,yscale = 4, shorten >= 0pt,
    shorten <= 0pt,]
\begin{scope}[xshift=15cm,yshift=-1.5cm]

\draw[Red,thick] (0,1/4) -- (1,5/6);
\draw[thick] (0,1/4) -- (1/2,5/8);
\draw[thick] (1/2,5/8) -- (1,5/6);

\draw[axis,->] (0,0) to (0,1) node[above] {$\cost(\mu)$}; 
\draw[axis,->] (0,0) to (1.05,0) node[right] {$\mu$};
				
\draw[axis] (1,0.05) -- (1,-0.05) node[below] {$1$};
\draw[axis] (0,0.05) -- (0,-0.05) node[below] {$0$};
\draw[axis] (1/2,0.05) -- (1/2,-0.05) node[below] {$\nicefrac{1}{2}$}; 

\draw[axis] (0,0) -- (-0.025,0) node[left] {$0$};
\draw[axis] (0.025,5/8) -- (-0.025,5/8) node[left] {$\nicefrac{5}{8}$};
\draw[axis] (0.025,1/4) -- (-0.025,1/4) node[left] {$\nicefrac{1}{4}$};
\draw[axis] (0.025,5/6) -- (-0.025,5/6) node[left] {$\nicefrac{5}{6}$};

\draw[ultra thick, Dandelion] (0,-0.3)--(1,-0.3);
\draw[ultra thick, MidnightBlue] (1/2,-0.35)--(1,-0.35); 

\end{scope}
\end{tikzpicture}
\end{center}
\caption{}
\end{subfigure}
\end{center}
\caption{A simple non-atomic congestion game:
(a) Instance with two $s$-$t$-paths. (b) The cost of the Wardrop equilibrium (black) and of the full information signaling scheme (red) depending on the probability $\mu=\P[d=1]$. As indicated by the colored lines underneath, the upper path is used for all $\mu\in[0,1]$ while the lower path is only used for $\mu\in[1/2,1]$.} 
\label{fig:example-intro}
\end{figure}

\begin{example}
\label{ex:firstExample}
Consider the simple game in \cref{fig:example-intro} with two parallel edges and cost functions $x$ and $5/6$, respectively. We normalize the highest demand to the size of the population and assume that the demands are $1$ and $1/2$, each with a probability of $1/2$. For a fixed player, let $A$ be the event that this player is active. If the player chooses the upper edge, their expected cost is given by the flow $x_1$ on the upper edge when they are active, and 0 otherwise, i.e., $\Pr[A] \cdot \mathbb{E}[x_1 \mid A] + \P[\neg A] \cdot 0$. Either all players are active or half of the players are active, and both cases have probability $1/2$. Hence, the expected cost amounts to the sum of $\P[A \mid d = 1/2] \cdot \P[d=1/2]\cdot \mathbb{E}[x_1 \mid A \wedge (d\!=\!1/2)]$ and $\P[A \mid d = 1]\cdot \P[d=1] \cdot \mathbb{E}[x_1 \mid A \wedge (d\!=\!1)]$. When a fraction $a$ of the players are active, we assume each infinitesimal player is active independently with probability $a$, i.e., $\P[A \wedge (d=a)] = a$, for every $a \in [0,1]$.
Recall that players choose their strategy before the demand is realized. Suppose a fraction of $x_1^*$ players chooses the upper edge. Then the expected cost of a player choosing the upper edge is $\smash{\tfrac{1}{2} \cdot \tfrac{1}{2} \cdot \tfrac{x_1^*}{2} + 1 \cdot \tfrac{1}{2} \cdot x_1^* = \tfrac{5}{8} \cdot x_1^*}$. A similar calculation shows that the expected cost when choosing the lower edge is $\smash{\left(\tfrac 12 \cdot \tfrac 12 + \tfrac 12 \cdot 1\right) \cdot \tfrac 56 = \tfrac 58}$. We conclude that in the unique Wardrop equilibrium $x_1^* = 1$, i.e., all players choose the upper edge. This leads to a total expected cost of $\smash{\tfrac 58 = 0.625}$.

Now suppose a traffic service provider like TomTom, Apple, or Google discloses a public signal whether the traffic is low (i.e., $d=1/2$) or high (i.e., $d=1$). We call this setting \emph{full information}, since the information about the state of the world is fully disclosed. Now every player updates their belief about the demand and conditions their route choice on this information. In both cases, a corresponding Wardrop equilibrium emerges. For demand $1/2$ we get $x_1^* = \tfrac 12$ and $x_2^* = 0$; for demand 1, $x_1^* = \tfrac 56$ and $x_2^* = \tfrac 16$. The total expected cost with full information is $\tfrac 12 \cdot \tfrac 12 \cdot \tfrac 12 + \tfrac 12  \cdot 1 \cdot \tfrac 56 = \tfrac{13}{24} \approx 0.542$. Hence, full information improves the total expected cost over the case with no signal.

Similarly, consider the no-signal case and the cost of an equilibrium as a function of any prior $\mu=\P[d=1] \in [0,1]$ (see the black function in \cref{fig:example-intro}). Inspecting the cost when the provider gives full information (the red function), we see that full information improves the cost over no signal, for any prior $\mu \in [0,1]$.
\hfill $\blacksquare$
\end{example}

In this paper, we study how to optimally reveal information about the realized demand in order to induce Wardrop equilibria with low total expected cost. We focus on the case of \emph{public signals} where the information provided is the same for all players. More specifically, we assume that a benevolent traffic service provider has a finite set of public and abstract signals $\Sigma$ at its disposal. Formally, a signal $\sigma \in \Sigma$ has no a-priori intrinsic meaning. In practice, however, signals may already be biased towards certain information, e.g., that traffic volume is ``moderate'', ``relatively high'', or ``gridlocked''. 
Before seeing the realized demand, the service provider commits to a signaling scheme that is public knowledge of all players. It fixes the probabilities of emitting the signals for each realization of the demand. Subsequently, the realization of the demand is observed by the service provider (e.g., due to traffic measurements or cell phone data) and signals are sent according to the predefined signaling scheme. Upon receiving a signal, the players update their beliefs about the realization of the demand and react by playing a corresponding Wardrop equilibrium. How can the service provider optimize the public signaling scheme in order to minimize the total expected cost of the induced Wardrop equilibrium?


\paragraph{Contribution.}

After introducing a formal description of the problem in \cref{sec:preliminaries}, we derive useful structural properties of equilibria and their cost functions in \cref{sec:struc-prop}. More precisely, we show that the cost function of the unique Wardrop equilibrium flow is piecewise linear in the belief of the realized state for any finite number of states (\Cref{lem:piecewise-linear}). We further show that for two different demands, the cost function is monotonically non-decreasing in the probability that the higher demand is realized (\Cref{lem:monotonicity} and Corollary~\ref{cor:2stateMonotone}). Building upon these properties, in \cref{sec:FPTAS} we provide a fully polynomial-time approximation scheme (FPTAS) for optimal signaling with two different states (\Cref{thm:fptas}).

There exist network structures, in which always revealing the realized state is an optimal signaling scheme, no matter which prior belief the players have (cf.\ Example~\ref{ex:firstExample} above). We call such a signaling scheme \emph{full information revelation}. In \cref{sec:sepa}, we show that if the underlying graph is a series-parallel graph, full information revelation is an optimal signaling scheme. We even prove that this characterisation is tight -- whenever the underlying graph is not series-parallel, there exist some cost functions for the edges and demands such that full information revelation is not optimal. 

In Sec.\ \ref{sec:LPs} we provide an LP-based algorithm that computes the optimal signaling scheme that induces Wardrop equilibria only using a distinct set of different supports. The algorithm works for any number of states and runs in time polynomial in the number of states, the number of edges, and the number of given supports. In general, however, there exist networks for which, over all beliefs, the number of different supports used in a Wardrop equilibrium is exponential in the input size, even when there are only two different demands. We conduct a computational study in \cref{ssec:compStudy} exhibiting that the number of different supports used in a Wardrop equilibrium is small on real-life instances. Therefore, our LP-based algorithm can be implemented in reasonable time in practice. Moreover, we see that the optimality of \emph{full information revelation} is ubiquitous in these instances, even though they are not series-parallel.

\paragraph{Related Work.}

The question how non-atomic network congestion games behave when the demand changes has been studied thoroughly. \citet{YounGJ08} and \citet{OhareCW16} examined empirically how the price of anarchy, i.e., the ratio of the total travel time of a Wardrop equilibrium and the total travel time of a system optimum, changes as a function of the demand. The functional dependence of the price of anarchy as a function of the demand has been studied analytically by \citet{Colini-Baldeschi19,Colini-Baldeschi20,CominettiDS19}, and \citet{WuMCX21}. \citet{WuMRX22} studied a similar question for atomic congestion games.
\citet{WangDC14} obtained bounds on the price of anarchy of Wardrop equilibria with stochastic demands depending on parameters of the distribution. \citet{CorreaHS19} studied a similar model with the difference that the players perform a Bayesian update of the distributions after observing whether their commodity travels and find that the price of anarchy transfers from the deterministic model; this refines earlier results of \citet{Roughgarden15} on the price of anarchy of Bayes-Nash equilibria. 
More generally, the sensitivity of Wardrop equilibria to changes in the demand was studied by \cite{Hall78,Josefsson07,Fisk79,EnglertFO10,Patriksson04,TakallooK20,KlimmW22,UkkusuriW10}.

Games with a random number of atomic players were introduced by \citet{Myerson1998}. He showed that when the distribution of the number of players follows a Poisson distribution, beliefs about the number of players of an internal player and an external observer coincide. Such Poisson games were further studied by \citet{Myerson00}.
\citet{GairingMT08} studied atomic congestion games where the weight of a player is their private information and provide bounds on the price of anarchy.
\citet{CominettiSSM22} studied Bernoulli congestion games, i.e., atomic congestion games where each player participates with an independent probability. They showed that the resource loads converge to a Wardrop equilibrium in the limit when the participation probability vanishes. \citet{CominettiSOM19} obtained bounds on the price of anarchy of Bernoulli congestion games with affine costs. Similar models where players participate only with a certain probability were studied in \cite{AngelidakisFL13,MeirTBK12}. \citet{AshlagiMT06} studied (non-Bayesian) congestion games with unknown number of players.

The potential of information design for non-atomic congestion games was illustrated through examples by \citet{DasKM17}. 
\citet{NachbarX21} further explored different signaling regimes and study connections with the price of anarchy.
\citet{MassicotL19} fully characterized the optimal policy for networks consisting of two edges with affine cost where the cost of one edge does not depend on the state. 
\citet{VassermanFH15} considered a setting with parallel edges with affine costs where the cost functions are permuted and bounded the improvements that can be obtained from private signals.
\citet{BhaskarCKS16} considered games with affine costs where the offset depends on the state and showed that the problem of computing an optimal signaling scheme cannot be approximated by a factor of $(4/3-\epsilon)$ for $\epsilon >0$, unless $\classP = \classNP$.
For the same setting, \citet{GriesbachHKK22} proved that revealing the realized state is always an optimal signaling scheme if and only if the underlying network is a series-parallel graph. They also provided LP-based techniques to compute the optimal signaling schemes. In particular, they can compute optimal signals for parallel links with a constant number of states and commodities.
\citet{acemoglu2018informational} considered the setting in which players have different knowledge about the available edges in a road network and give a strict characterization of the graph class for which a player cannot obtain higher private cost by gaining additional information. 
\citet{WuAO21} characterize the Bayesian Wardrop equilibria that arise when populations of drivers receive multiple signals from heterogeneous information systems.
\citet{ZhouNX22} showed how to compute the optimal public and private signal in an atomic congestion game with constant number of parallel edges.
\citet{CastiglioniCM021} studied information design for atomic congestion games in the relaxed setting of ex ante persuasion where the players are only persuaded to follow the signaling scheme \emph{before} receiving the signal. They showed that an optimal signal can be computed with LP-based techniques for symmetric players, and show that the problem is $\mathsf{NP}$-hard for asymmetric players. The provision of information in a dynamic model where players have preferences over arrival times was explored by \citet{ArnottPL91}.

\section{Preliminaries}
\label{sec:preliminaries}

\paragraph{Signaling.} 
We consider a signaling problem in the context of network congestion games. There is a finite set $\Theta = \{\theta_1,\dots,\theta_l\}$ of $l$ \emph{states of nature}, along with a prior distribution $\prior$, where $\prior_\theta \ge 0$ is the probability that state $\theta \in \Theta$ is realized. We denote by $\Delta(\Theta)$ the space of all distributions over $\Theta$.	
There is a sufficiently large, finite set of (public) \emph{signals} $\Sigma$ which can be used by a benevolent principal to influence the information of all players in a congestion game. We study the problem of computing a good signaling scheme, given by a distribution over $\Sigma$ for each state $\theta \in \Theta$. More formally, a \emph{signaling scheme} is a matrix $\phi = (\phi_{\theta,\sigma})_{\theta \in \Theta, \sigma \in \Sigma}$ such that $\phi_{\theta,\sigma} \ge 0$ for all $\theta \in \Theta, \sigma \in \Sigma$, and $\sum_{\sigma \in \Sigma} \phi_{\theta,\sigma} = \prior_\theta$ for each $\theta \in \Theta$. The value of $\phi_{\theta,\sigma}$ is the combined probability that state~$\theta$ is realized and the sender sends signal $\sigma$. We define $\phi_\sigma = \sum_{\theta \in \Theta} \phi_{\theta,\sigma}$ as the total probability that signal $\sigma$ gets sent. A signal $\sigma \in \Sigma$ \emph{gets issued} by scheme $\phi$ if $\phi_{\sigma} > 0$.

The scenario proceeds as follows. First, the principal commits to a signaling scheme $\phi$ and communicates this to all players. Hence, the prior $\prior$ and the signaling scheme $\phi$ are public knowledge. Then the state of nature is realized. The principal sees the realized state $\theta$ and sends a public signal $\sigma$ chosen according to $\phi$. All agents receive the signal, update their beliefs about the state of nature $\theta$ and the resulting costs in the congestion game, and then choose equilibrium strategies as a result of (unilaterally) minimizing their individual expected cost. The goal of the principal is to choose $\phi$ to minimize the total expected cost of the resulting equilibrium.

\paragraph{Network Congestion Games.}
We now describe the network congestion game, the (individual) expected cost of the agents, and the total expected cost. There is a directed graph $G=(V,E)$ with a designated source $s \in V$ and destination $t \in V$.
For every edge~$e \in E$ there is a \emph{cost function} $c_{e}: \R_{\geq 0} \to \R_{\geq 0}$ that is convex and non-decreasing. In this paper, we focus on \emph{affine costs}, i.e., all functions $c_e$ are of the form $c_e(x) = a_e  x + b_e$, where $a_e \in \mathbb{R}_{>0}$ and $b_e \in \mathbb{R}_{\geq 0}$ for every~$e \in E$.
We concentrate on single-commodity games, in which all players want to route from $s$ to $t$. The player population consists of a continuum of infinitesimally small players of total volume $d > 0$, the \emph{available demand}. For simplicity, we normalize to $d=1$. The \emph{actual demand}, however, is uncertain and depends on the realized state of nature. Formally, each state $\theta \in \Theta$ is associated with a realization $d_\theta \le d$ of actual demand. We assume w.l.o.g.\ $0 < d_{\theta_1} < d_{\theta_2} < \dots < d_{\theta_l} = d = 1$. Intuitively, in state $\theta$, each infinitesimal player participates in the game independently with probability $d_\theta/d = d_\theta$.

The set of feasible strategies $\calP \subseteq 2^E$ for each player is the set of directed $s$-$t$-paths in $G$. A \emph{strategy distribution} or \emph{path flow} is a distribution of the players on the paths $P \in \calP$. Such a path flow is represented by a vector $x = (x_P)_{P \in \calP}$ satisfying the three properties (1) $\sum_{P \in \calP} x_P = 1$, (2) $x_P \geq 0$ for all $P \in \calP$, and (3) $x_P = 0$ for all $P \notin \calP$. Let $\calX$ denote the set of those vectors. Every path flow $x \in \calX$ induces a \emph{load} $x_e$ on every edge~$e \in E$ given by  $x_e = \sum_{P \in \calP : e \in P}   x_{P}$.

Upon receiving the public signal $\sigma$ (and knowing $\phi$), all players perform a Bayes update of $\prior$ to a conditional distribution $\mu_{\sigma} \in \Delta(\Theta)$. The conditional probability of $\theta \in \Theta$ when receiving a signal $\sigma$ with $\phi_{\sigma} > 0$ is given by $\mu_{\theta,\sigma} = \phi_{\theta,\sigma} / \phi_\sigma$. 
Indeed, every signaling scheme $\phi$ can be seen as a convex decomposition of $\prior$ into distributions $\mu_\sigma$, i.e., for every $\theta \in \Theta$, $\prior_\theta = \sum_{\sigma} \phi_{\theta,\sigma} = \sum_{\sigma} \phi_{\sigma}  \mu_{\theta,\sigma}$.
After this update, the players choose a path flow $x \in \calX$. The next definitions apply for every distribution $\mu \in \Theta(\Delta)$. Given $\mu$, suppose a player chooses a path including edge $e$. In state $\theta$, the player is present in the system with probability $d_{\theta}$, otherwise the private cost of this player is 0. Conditioned on the presence of this (infinitesimal) player, the expected cost that player will experience on $e$ is $c_e(d_\theta  x_e)$. Hence, the \emph{expected cost of an edge} $e \in E$ is $ c_e(x_e \mid \mu) = \sum_{\thetasum}  \mu_\theta  d_\theta  c_e \left( d_\theta x_e \right)$.
The \emph{(individual) expected cost of a player} on path $P \in \calP$ is given by $c_P(x \mid \mu) = \sum_{e \in P} c_e(x_e \mid \mu)$.
A path flow $x \in \calX$ is a \emph{Wardrop equilibrium} if no player has an incentive to change their chosen strategy. Formally, given $\mu$, no player shall improve their expected cost by deviating to another strategy, i.e., $c_P(x \mid \mu) \leq c_Q(x \mid \mu)$ for all $P,Q \in \calP$ with $x_P > 0$.

The next result extends a characterization for non-atomic games with a single state, cf.~\cite{BeckmannMW56}. It directly carries over to the scenario considered in this paper as follows.
\begin{proposition}
	\label{pro:beckmann}
	Given any distribution $\mu \in \Delta(\Theta)$, a strategy distribution $x \in \mathcal{X}$ is a Wardrop equilibrium if and only if $x \in \arg\min \left\{\sum_{e \in E} \int_{0}^{y_e} c_{e}(t \mid \mu)\,\mathrm{d}t \;:\; y \in \mathcal{X}\right\}$.
\end{proposition}
All cost functions $c_e$ are convex, so their convex combinations $c_{e}(x \mid \mu)$ are convex as well. As such, the Wardrop equilibrium is unique, since the optimization problem in Proposition~\ref{pro:beckmann} is strictly convex and has a unique solution. We use $x^*(\mu)$ to denote the unique Wardrop equilibrium for a distribution $\mu \in \Delta(\Theta)$. The total cost of a path flow $x$ for $\mu \in \Delta(\Theta)$ is given by $  \cost(x \mid \mu) = \sum_{P \in \calP} x_P  c_P(x \mid \mu) = \sum_{e \in E} x_e c_e(x_e \mid \mu)$.
For the Wardrop equilibrium for $\mu \in \Delta(\Theta)$, we use the short notation $\cost(\mu) = \cost(x^*(\mu) \mid \mu)$. The goal of the principal is to choose $\phi$ in order to minimize the total expected cost of the Wardrop equilibrium for the conditional distributions $\mu_\sigma$ resulting from all signals $\sigma$, i.e.,
$\cost(\phi) = \sum_{\sigma \in \Sigma} \phi_{\sigma}  \cost(\mu_\sigma).$ 
We refer to \Cref{app:prelim} for a more detailed example illustrating the problem and its concepts.

\section{Structural Properties}
\label{sec:struc-prop}
We exhibit useful structural properties of the signaling scenario outlined above. We concentrate on a single probability distribution $\mu$ over states of nature, i.e., $\mu_{\theta}$ is the probability of state $\theta$ (i.e., that demand $d_\theta$ in the network is realized). Given $\mu$, the expected cost of edge $e \in E$ is
\begin{align}
\label{eq:exp-edge-cost}
c_e ( x_e \mid \mu ) &= \sum_{\thetasum} \mu_\theta  d_\theta  c_e (d_\theta x_e) =a_e \sum_{\thetasum} \mu_\theta^{\phantom{2}} d_\theta^2  x_e + b_e\sum_{\thetasum} \mu_\theta d_\theta.    
\end{align}
We first show that the cost of the unique Wardrop equilibrium with respect to $c_e (\cdot \mid \mu)$ is piecewise linear in $\mu\in \Delta(\Theta)$. 
Let $x$ be a flow. For $v\in V$, let $\psi_v$ be the length of a shortest path with respect to $c_e(x_e \mid \mu )$ from $s$ to $v$. We call an edge $e=(v,w)$ \emph{active} in $x$ if $\psi_w-\psi_v=c_e(x_e\mid \mu)$. Clearly, for every flow $x$, the set of active edges is connected and such that every vertex $v$ is reached by a path of active edges from $s$. Let
$\mathcal{A}=\{A\subseteq E\mid G=(V,A)$ is connected and contains an $(s,v)$-path for all $v\in V\}$
be the set containing all sets of edges with that property, and let $A(x)$ be the set of active edges for a flow $x$. In the following, we call a set $A\in \mathcal{A}$ a \emph{support}.

\begin{lemma}
\label{lem:piecewise-linear}
For a single-commodity network congestion game, the unique Wardrop equilibrium flow and
the cost of the unique Wardrop equilibrium are piecewise linear in $\mu$. 
In particular, for every $A \in \mathcal{A}$, there is a possibly empty polytope $P_A \subseteq \Delta(\Theta)$ such that $P_A = \{\mu \in \Delta(\Theta) \mid A(x^*(\mu)) = A\}$, and $x^*$ and $C$ are affine on $P_A$.
\end{lemma}

\begin{proof}
Define a balance vector $\beta = (\beta_v)_{v \in V}$ as
\begin{align}
\label{eq:balance}
\beta_v &= \begin{cases}
  -\sum_{\thetasum}\mu^{\phantom{2}}_\theta d_\theta^2 & \text{ if } v=s,\\
 \phantom{-}\sum_{\thetasum}\mu^{\phantom{2}}_\theta d_\theta^2  & \text{ if } v=t,\\
 \phantom{-}0 & \text{ otherwise.}
 \end{cases}
\end{align}
Let $x^*$ be the unique Wardrop equilibrium, i.e., the optimal solution to the optimization problem
\begin{align*}
\text{Min.}  &\sum_{e \in E} \int_{0}^{x_e} c_e(z \mid \mu) \;\text{d}z\\
\text{s.t. }  &\sum_{e \in \delta^+(v)} \sum_{\thetasum}\mu_\theta^{\phantom{2}} d_\theta^2 x_e - \sum_{e \in \delta^-(v)} \sum_{\thetasum}\mu_\theta^{\phantom{2}} d_\theta^2 x_e = \beta_v
&& \text{ for all } v \in V,\\
&x_e \geq 0 &&\text{ for all } e \in E. 
\end{align*}
By the Karush-Kuhn-Tucker optimality conditions, a flow $x = (x_e)_{e \in E}$ is optimal if and only if it is feasible and there is a dual vector $\pi = (\pi_v)_{v \in V}$ such that
\begin{align*}
c_e(x_e \mid \mu) &= \pi_w - \pi_v &&\text{ for all } e = (v,w) \in E \text{ with } x_e \neq 0,\\
c_e(x_e \mid \mu) &\geq \pi_w - \pi_v &&\text{ for all } e = (v,w) \in E \text{ with } x_e = 0.
\end{align*}
The dual variables are clearly invariant under additive shifts. Thus it is without loss of generality to assume that $\pi_s = 0$. Using~\eqref{eq:exp-edge-cost}, we conclude that there is a support $A \in \mathcal{A}$ such that a Wardrop equilibrium satisfies the following equations

\begin{subequations}
\label{eq:wardrop-equation}
\begin{align}
\pi_v + a_e \,\left( \sum_{\thetasum}\mu_\theta^{\phantom{2}} d_\theta^2  \right) \, x_e +  b_e \sum_{\thetasum}\mu_\theta d_\theta  &= \pi_w && \text{ for all } e \in A,\\
\sum_{e \in \delta^+(v)} \sum_{\thetasum}\mu_\theta^{\phantom{2}} d_\theta^2 x_e - \sum_{e \in \delta^-(v)} \sum_{\thetasum}\mu_\theta^{\phantom{2}} d_\theta^2 x_e &= \beta_v && \text{ for all } v \in V,\\
\pi_s &= 0.
\end{align}
\end{subequations}
as well as the inequalities
\begin{subequations}
\label{eq:wardrop-inequality}
\begin{align}
\pi_v + a_e \,\left( \sum_{\thetasum}\mu_\theta^{\phantom{2}} d_\theta^2  \right) \, x_e +  b_e \sum_{\thetasum}\mu_\theta d_\theta &\geq \pi_w && \text{ for all } e \in E \setminus A,\label{eq:wardrop-inequality-1}\\
x_e &\geq 0 && \text{ for all } e \in E.\label{eq:wardrop-inequality-2}
\end{align}
\end{subequations}

We claim that for all $A \in \mathcal{A}$, the linear system \eqref{eq:wardrop-equation} has full rank. To see this, we substitute $y_e = \left(\sum_{\thetasum}\mu_\theta^{\phantom{2}} d_\theta^2\right)x_e$ and we let $\Gamma \in \mathbb{R}^{V \times E}$ be the incidence matrix of the subgraph $G[A]$,
i.e., $\Gamma = (\gamma)_{v \in V, e \in E}$ defined as $\gamma_{v,e} = 1$, if $e \in \delta^-(v)$, $\gamma_{v,e} = -1$ if $e \in \delta^+(v)$, and $\gamma_{v,e} = 0$, otherwise. Let $\smash{D \in \mathbb{R}^{E \times E} = \text{diag}(a_1,\dots,a_{k})}$ with $k = |A|$ be the diagonal matrix with the slopes of the cost function of the edges $e\in A$ on the diagonal. Eliminating $\pi_s = 0$, the system \eqref{eq:wardrop-equation} can be written as
\begin{align}
\label{eq:matrix-equation}
\left[
\begin{array}{c c}
D & \hat{\Gamma}^\top\\
\hat{\Gamma} & \mathbf{0}	
\end{array}
\right]
\left[
\begin{array}{c}
y \\
\hat{\pi}
\end{array}
\right] =
\left[
\begin{array}{c}
- b_e \, \sum_{\thetasum}\mu_\theta d_\theta\\
\hat{\beta} 
\end{array}
\right],
\end{align}
where $\hat{\pi}$ is the vector of vertex potentials with the entry for $s$ removed, $\hat{\Gamma}$ is the incidence matrix with the row for $s$ removed, and $\hat{\beta}$ is the vector $\beta$ with the entry for $s$ removed. Using Schur complements, we obtain that the matrix on the left hand side of \eqref{eq:matrix-equation} is invertable if and only if $\hat{L} := \hat{\Gamma} D^{-1} 	\hat{\Gamma}^\top$ is invertible; see, e.g., the textbook by Harville~\cite[Theorem~8.5.11]{harville1997matrix} in which case the inverse is given by
\begin{align} \label{eq:harville}
\left[
\begin{array}{c c}
D & \hat{\Gamma}^\top\\
\hat{\Gamma} & \mathbf{0}	
\end{array}
\right]^{-1} \!= 
\left[
\begin{array}{c c}
\!D^{-1} - D^{-1}\hat{\Gamma}^\top \hat{L}^{-1} \hat{\Gamma} D^{-1}  &  \!D^{-1} \hat{\Gamma}^\top \hat{L}^{-1}\\
\!\hat{L}^{-1} \hat{\Gamma} D^{-1} & \!-\hat{L}^{-1}
\end{array}
\right].
\end{align}

The matrix $\hat{L}$ is a weighted Laplacian matrix of a connected graph (with the entry for $s$ removed) which is known to have full rank. This implies that for fixed $\mu$, there is a unique solution $x$ satisfying \eqref{eq:wardrop-equation} and a unique value for $\pi_t$. 

The equations~\eqref{eq:wardrop-equation} and the inequalities~\eqref{eq:wardrop-inequality} together with $\mu \in [0,1]^l$ define a polytope of all vectors $(\mu,x,\pi) \in \mathbb{R}^{\Theta \times E \times V}$ such that $\mu \in [0,1]^l$ is a distribution, $x$ is a corresponding Wardrop equilibrium with support $A$, and $\pi$ is a corresponding vector of vertex potentials. Since the projection of a polytope is a polytope again, the set $P_A$ is a polytope as well.

Observe that \eqref{eq:wardrop-equation} gives a system of linear equations which, for fixed $\mu$, have a unique solution in $x$. Thus, the Wardrop equilibrium $x^*(\mu)$ is an affine function in $\mu$ on $P_A$. To see that also the cost of the Wardrop equilibrium is affine on $P_A$ note that the cost of the Wardrop equilibrium is given by $d\pi_t$ and, hence, the result follows.
\end{proof}

Next, we study conditions under which the cost $\cost(\mu)$ is monotone in $\mu$. 

\begin{lemma}
\label{lem:monotonicity}
Let $\mu^{(1)}, \mu^{(2)} \in \Delta(\Theta)$ be such that we have 
\begin{align*}
\frac{\sum_{\thetasum} \mu^{(1)}_\theta d_\theta^2}{\sum_{\thetasum} \mu^{(1)}_\theta d_\theta^{\phantom{)}}} &< \frac{\sum_{\thetasum} \mu^{(2)}_\theta d_\theta^2}{\sum_{\thetasum} \mu^{(2)}_\theta d_\theta^{\phantom{)}}} &&\text{and} & \sum_{\thetasum} \mu^{(1)}_\theta d_\theta^{\phantom{)}} &< \sum_{\thetasum} \mu^{(2)}_\theta d_\theta^{\phantom{)}}.
\end{align*}
Then $C(\mu^{(1)}) \leq C(\mu^{(2)})$. 
\end{lemma}

\begin{proof}
Let $x$ be the Wardrop equilibrium with respect to $\mu$, i.e., 
\begin{align*}
\sum_{e \in P} c_e(x_e \mid \mu) &\leq \sum_{e \in Q} c_e(x_e \mid \mu) \quad \text{ for all  $P,Q \in \calP$ with $x_P > 0$,}\\
\sum_{P \in \calP} x_P &= 1.
\end{align*}
Using the definition of $c_e(x_e \mid \mu)$ this yields
\begin{align*}
\sum_{e \in P} \Biggl[a_e\left(\sum_{\thetasum} \mu_\theta^{\phantom{2}} d_\theta^2\right) x_e + b_e \sum_{\thetasum} \mu_\theta d_\theta\Biggr] &\leq \sum_{e \in Q} \left[ a_e\left(\sum_{\thetasum} \mu_\theta^{\phantom{2}} d_\theta^2\right)x_e  + b_e \sum_{\thetasum} \mu_\theta d_\theta \right] \text{ for all $P,Q \in \mathcal{P}$ with $x_P > 0$,}\\
\sum_{P \in \mathcal{P}} x_P &= 1.
\end{align*}
Substituting $y_e := (\sum_{\thetasum} \mu_\theta^{\phantom{2}} d_\theta^2) x_e$ and $y_P = \sum_{e \in P} y_e$, we obtain
\begin{align*}
\sum_{e \in P} \left(a_e y_e  + b_e \sum_{\thetasum} \mu_\theta d_\theta\right) &\leq \sum_{e \in Q} \left(a_e y_e  + b_e \sum_{\thetasum} \mu_\theta d_\theta \right) \text{ for all $P,Q \in \calP$ with $y_P > 0$,}\\
\sum_{P \in \mathcal{P}} y_P &= \sum_{\thetasum} \mu_\theta^{\phantom{2}} d_\theta^2.
\intertext{Further substituting $z_e := y_e \big/ (\sum_{\thetasum} \mu_\theta d_\theta)$ and $z_P = \sum_{e \in P} z_e$, we obtain}
\sum_{e \in P} ( a_e z_e  + b_e ) &\leq \sum_{e \in Q} (a_e z_e  + b_e ) \text{ for all $P,Q \in \mathcal{P}$ with $z_P > 0$,}\\
\sum_{P \in \mathcal{P}} z_P &= \tfrac{\sum_{\thetasum} \mu_\theta^{\phantom{2}} d_\theta^2}{\sum_{\thetasum} \mu_\theta d_\theta}.
\end{align*}
Hence, we observe that $z$ is a Wardrop equilibrium for a deterministic ``virtual'' demand of $\smash{\bigl(\sum_{\thetasum} \mu_\theta^{\phantom{2}} d_\theta^2\bigr)\big/\left(\sum_{\thetasum} \mu_\theta d_\theta\right)}$. In single-commodity network congestion games with affine linear costs, the per-unit cost of a Wardrop equilibrium is known to be non-decreasing in the demand~\cite{KlimmW22}. 

Now consider the distributions $\mu^{(1)}$ and $\mu^{(2)}$ along with the induced Wardrop equilibria $x^{(1)}$ and $x^{(2)}$. We define the substituted quantities $y_e^{(i)}$ and $z_e^{(i)}$ in the obvious way, for $i=1,2$. Let $P^{(i)} \in \mathcal{P}$ be a path such that $x^{(i)}_{P^{(i)}} > 0$, for each $i=1,2$. Then, we obtain
\begin{align*}
\sum_{e \in P^{(1)}} ( a_e^{\phantom{)}} z^{(1)}_e + b_e) \leq \sum_{e \in P^{(2)}} ( a_e^{\phantom{)}} z^{(2)}_e + b_e ).
\end{align*}
Substituting back, this implies
\begin{align*}
&\tfrac{1}{\sum_{\thetasum} \mu^{(1)}_\theta d_\theta^{\phantom{)}}}\sum_{e \in P^{(1)}} ( a_e^{\phantom{)}} y^{(1)}_e + b_e) \leq \tfrac{1}{\sum_{\thetasum} \mu^{(2)}_\theta d_\theta^{\phantom{)}}}\sum_{e \in P^{(2)}} ( a_e^{\phantom{)}} y^{(2)}_e + b_e),
\end{align*}
and further
\begin{align*}
\tfrac{1}{\sum_{\thetasum} \mu^{(1)}_\theta d_\theta^{\phantom{)}}}\sum_{e \in P^{(1)}} \left[ a_e  \bigl(\sum_{\thetasum} \mu_\theta^{(1)} d_\theta^2\bigr) x^{(1)}_e + b_e \right] \leq
\tfrac{1}{\sum_{\thetasum} \mu^{(2)}_\theta d_\theta^{\phantom{)}}}\sum_{e \in P^{(2)}} \left[ a_e \bigl(\sum_{\thetasum} \mu_\theta^{(2)} d_\theta^2\right) x^{(2)}_e + b_e\bigr].
\end{align*}
Multiplying this inequality with $\sum_{\thetasum} \mu_\theta^{(1)} d_\theta^{\phantom{)}}$ yields
\begin{align}
\begin{split}
\label{eq:wardrop-inequality-proof}
\sum_{e \in P^{(1)}} \left[ a_e \left(\sum_{\thetasum} \mu_\theta^{(1)} d_\theta^2\right)x^{(1)}_e + \left(\sum_{\thetasum} \mu_\theta^{(1)} d_\theta^{\phantom{)}}\right)b_e \right] &\leq \tfrac{\sum_{\thetasum}\mu_\theta^{(1)} d_\theta^{\phantom{)}}}{\sum_{\thetasum} \mu^{(2)}_\theta d_\theta^{\phantom{)}}}
 \!\!\sum_{e \in P^{(2)}}  \!\left[ a_e \!\left(\sum_{\thetasum} \mu_\theta^{(2)} d_\theta^2\right) x^{(2)}_e \! + \! \left(\sum_{\thetasum} \mu_\theta^{(1)} d_\theta^{\phantom{(}}\right)b_e \right]\\
&\leq \sum_{e \in P^{(2)}} \left[ a_e \left(\sum_{\thetasum} \mu_\theta^{(2)} d_\theta^2\right)x^{(2)}_e + \left(\sum_{\thetasum} \mu_\theta^{(2)} d_\theta^{\phantom{)}}\right)b_e\right].
\end{split}
\end{align}
Since $x^{(i)}$ is a Wardrop equilibrium with $\sum_{P \in \calP} x_P^{(i)} = 1$ and $x^{(i)}_{P^{(i)}} > 0$, we have
\begin{align*}
C(\mu^{(i)}) 
&= \sum_{e \in P^{(i)}}\left[ a_e \left(\sum_{\thetasum} \mu_\theta^{(i)} d_\theta^2\right)x^{(i)}_e + \left(\sum_{\thetasum} \mu_\theta^{(i)} d_\theta^{\phantom{(}}\right)b_e \right]\\
&= \sum_{e \in P^{(i)}} c_e(x_e \mid \mu^{(i)}),
\end{align*}
and the result follows from \eqref{eq:wardrop-inequality-proof}.
\end{proof}

As an immediate corollary, we obtain that the function $C$ is monotonic for the case that $|\Theta|=2$.

\begin{corollary}
\label{cor:2stateMonotone}
If $|\Theta|=2$, $C(\mu)$ is non-decreasing in $\mu_{\theta_2}$.
\end{corollary}

\begin{proof}[Proof of \Cref{cor:2stateMonotone}]
We apply \Cref{lem:monotonicity} by showing that the two assumptions hold.
Let $\mu^{(1)},\mu^{(2)}\in \Delta(\Theta)$ be such that $\mu_{\theta_2}^{(1)}<\mu_{\theta_2}^{(2)}$. 
The first assumption holds true by a series of equivalence transformations:
\begin{align*}
    &\frac{\sum_{\thetasum} \mu^{(1)}_\theta d_\theta^2}{\sum_{\thetasum} \mu^{(1)}_\theta d_\theta^{\phantom{(}}}
    <\frac{\sum_{\thetasum} \mu^{(2)}_\theta d_\theta^2}{\sum_{\thetasum} \mu^{(2)}_\theta d_\theta^{\phantom{(}}} \\
    \Leftrightarrow \quad 
    &\frac{\Bigl(1-\mu_{\theta_2}^{(1)}\Bigr)d_{\theta_1}^2+\mu_{\theta_2}^{(1)}}{\Bigl(1-\mu_{\theta_2}^{(1)}\Bigr)d_{\theta_1}^{\phantom{(}}+\mu_{\theta_2}^{(1)}}
    <\frac{\Bigl(1-\mu_{\theta_2}^{(2)}\Bigr)d_{\theta_1}^2+\mu_{\theta_2}^{(2)}}{\Bigl(1-\mu_{\theta_2}^{(2)}\Bigr)d_{\theta_1}^{\phantom{(}}+\mu_{\theta_2}^{(2)}}\\
    \Leftrightarrow \quad 
    &d_{\theta_1}^3\left( 1- \mu_{\theta_2}^{(1)}\right) \left(1-\mu_{\theta_2}^{(2)}\right)
    +\mu_{\theta_2}^{(1)}\mu_{\theta_2}^{(2)}\left( 1-d_{\theta_1}^{\phantom{2}}-d_{\theta_1}^2\right) +d_{\theta_1}^{\phantom{(}}\mu_{\theta_2}^{(1)}
    +d_{\theta_1}^2\mu_{\theta_2}^{(2)}\\
    & < d_{\theta_1}^3\left( 1- \mu_{\theta_2}^{(1)}\right) (1-\mu_{\theta_2}^{(2)})
    +\mu_{\theta_2}^{(1)}\mu_{\theta_2}^{(2)}( 1-d_{\theta_1}^{\phantom{2}}-d_{\theta_1}^2)+d_{\theta_1}^{\phantom{(}}\mu_{\theta_2}^{(2)}
    +d_{\theta_1}^2\mu_{\theta_2}^{(1)}\\
    \Leftrightarrow \quad 
    &\mu_{\theta_2}^{(1)}(1-d_{\theta_1}^{\phantom{(}}) <\mu_{\theta_2}^{(2)}(1-d_{\theta_1}^{\phantom{(}})\\
    \Leftrightarrow \quad &\mu_{\theta_2}^{(1)} <\mu_{\theta_2}^{(2)}
\end{align*}
For the last transformation we used the fact that $1-d_{\theta_1}>0$. 
The second assumption follows from the following computation:
\begin{align*}
    \sum_{\thetasum}\mu_\theta^{(1)}d_\theta^{\phantom{(}}
    &=\Bigl(1-\mu_{\theta_2}^{(1)}\Bigr)d_{\theta_1}^{\phantom{(}}+\mu_{\theta_2}^{(1)}\\
    &=d_{\theta_1}^{\phantom{(}}+\mu_{\theta_2}^{(1)}(1-d_{\theta_1})\\
    &<d_{\theta_1}^{\phantom{(}}+\mu_{\theta_2}^{(2)}(1-d_{\theta_1}^{\phantom{(}})\\
    &=\Bigl(1-\mu_{\theta_2}^{(2)}\Bigr)d_{\theta_1}^{\phantom{(}}+\mu_{\theta_2}^{(2)}\\
    &=\sum_{\thetasum}\mu_\theta^{(2)}d_\theta^{\phantom{(}}.
\end{align*}
Since both assumptions of \Cref{lem:monotonicity} hold, we obtain $C\left(\mu^{(1)}\right) \leq C\left(\mu^{(2)}\right)$, as desired.
\end{proof}

The proof of~\Cref{lem:piecewise-linear} has striking similarities to the proof of the same result for the model with (known demand and) affine costs and \emph{uncertain offsets} in \cite[Lemma~1]{GriesbachHKK22}. We have not been able to derive a direct reduction between the two scenarios and discuss why it seems non-obvious to establish.
First, \Cref{lem:monotonicity} and Corollary \ref{cor:2stateMonotone} do not hold for signaling with uncertain offsets. In more detail, reinspecting the proof of \Cref{lem:monotonicity}, we can reinterpret our model using deterministic demand $d=1$ and affine costs with \emph{uncertain slopes and offsets}
$
c_e ( x_e \mid \mu ) = \sum_{\thetasum} \mu_\theta  \left(a_e^{\phantom{2}} d_\theta^2  x_e +  b_e d_\theta\right)
= \sum_{\thetasum} \mu_\theta \left(a_{e,\theta}  x_e + b_{e,\theta}\right).
$
This scenario has been studied in, e.g.,~\cite{BhaskarCKS16,DasKM17}. The reinterpretation per se does not appear to be very useful -- games with uncertain affine costs are not very well-understood and in general do not admit, e.g., the linearity properties of~\Cref{lem:piecewise-linear} (in contrast to the case when only offsets are uncertain, cf.~\cite{GriesbachHKK22}). For a normalized version of the costs
$
    \smash{c_e^n(x_e \mid \mu)} =
    \smash{\bigl(\sum_{\thetasum} \mu_\theta  d_\theta  c_e(d_\theta x_e) \bigr)\big/\bigl(\sum_{\thetasum} \mu_\theta^{\phantom{2}} d_\theta^2\bigr)}
    = \smash{a_e x_e + b_e \bigl(\sum_{\thetasum} \mu_\theta d_\theta\bigr)\big/\bigl(\sum_{\thetasum} \mu_\theta^{\phantom{2}} d_\theta^2\bigr)},
$ 
a fixed $\mu$ yields the same scaling factor throughout for every edge cost. As such, every Wardrop equilibrium w.r.t.\ costs $c_e(\cdot \mid \mu)$ is also a Wardrop equilibrium w.r.t.\ costs $c_e^n(\cdot \mid \mu)$ and vice versa. The normalized costs $c_e^n$ indeed might seem like a reduction to an instance of affine costs with \emph{uncertain offsets}. However, defining state-specific constants $b_{e,\theta}^n$ independent of $\mu$ such that $\smash{\sum_{\thetasum} \mu_\theta^{\phantom{n}} b_{e,\theta}^n = b_e \bigl(\sum_{\thetasum} \mu_\theta d_\theta\bigr) \,\big/\, \bigl(\sum_{\thetasum} \mu_\theta^{\phantom{2}} d_\theta^2 \bigr)}$ for every $\mu \in \Delta(\Theta)$ can be impossible. This reduction would be \emph{non-linear} and, as such, substantially change the cost structure of signaling schemes.

\section{FPTAS for Two States}
\label{sec:FPTAS}

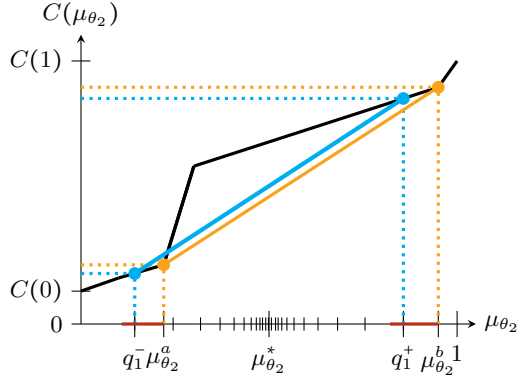
\begin{figure}[t]
		\centering
\footnotesize
 \begin{tikzpicture}
 [xscale=5,yscale =3.5,    shorten >= 0pt,
    shorten <= 0pt,]
\begin{scope}[xshift=15cm,yshift=-1.5cm]

\draw[axis,->] (0,0) to (0,1.1) node[above] {$\cost(\mu_{\theta_2})$}; 
\draw[axis] (0,0) -- (-0.025,0) node[left] {$0$};
\draw[axis] (0.025,1) -- (-0.025,1) node[left] {$\cost(1)$};
\draw[axis] (0.025,0.125) -- (-0.025,0.125) node[left] {$\cost(0)$};

\draw[axis,->] (0,0) to (1.05,0) node[right] {$\mu_{\theta_2}$};
\draw[axis] (1,0.05) -- (1,-0.05) node[below] {$1$};
\draw[axis] (0.5,0.05) -- (0.5,-0.05) node[below] {$\prior_{\theta_2}$};
\draw[axis] (0.22,0.05) -- (0.22,-0.05) node[below] {$\mu^a_{\theta_2}$};
\draw[axis] (0.14285714285714285,0.05) -- (0.14285714285714285,-0.05) node[below] {$q^-_1$};
\draw[axis] (0.95,0.05) -- (0.95,-0.05) node[below] {$\mu^b_{\theta_2}$};
\draw[axis] (0.8571428571428572,0.05) -- (0.8571428571428572,-0.05) node[below] {$q^+_1$};

\foreach \x in {0.14285714285714285,0.2448979591836734, 0.3177842565597667, 0.3698458975426905, 0.4070327839590646, 0.4335948456850462, 0.45256774691789015, 0.46611981922706436, 0.4757998708764746, 0.4827141934831961, 0.49118071096081434,0.49836019050113317}
 \draw[axis] (\x,0.025) -- (\x,-0.025) node[below] {};

\foreach \x in {0.8571428571428572, 0.7551020408163266, 0.6822157434402333, 0.6301541024573095, 0.5929672160409354, 0.5664051543149539, 0.5474322530821099, 0.5338801807729356, 0.5242001291235254, 0.5172858065168039,0.5062994921708469}
 \draw[axis] (\x,0.025) -- (\x,-0.025) node[below] {};




\draw[very thick] (0,0.125) -- (0.1,0.175);
\draw[very thick] (0.1,0.175) -- (0.22,0.225);

\draw[very thick] (0.22,0.225) -- (0.3,0.6);
\draw[very thick]  (0.3,0.6) -- (0.22,0.225);
\draw[very thick] (0.3,0.6) -- (0.95,0.9);
\draw[very thick] (0.95,0.9) -- (1,1);


\draw[dotted,YellowOrange,very thick] (0.22,0) -- (0.22,0.225);
\draw[dotted,YellowOrange,very thick] (0.95,0) -- (0.95,0.9);

\draw[dotted,YellowOrange,very thick] (0,0.225) -- (0.22,0.225);
\draw[dotted,YellowOrange,very thick] (0,0.9) -- (0.95,0.9);

\draw[ultra thick,Cyan] (0.14285714285714,0.192) -- (0.8571428571428572,0.8575);

\draw[dotted,Cyan,very thick] (0.14285714285714285,0) -- (0.14285714285714285,0.192);
\draw[dotted,Cyan,very thick] (0,0.192) -- (0.14285714285714285,0.192);

\draw[dotted,Cyan,very thick] (0.8571428571428572,0) -- (0.8571428571428572,0.8575);
\draw[dotted,Cyan,very thick] (0,0.8575) -- (0.8571428571428572,0.8575);

\draw[YellowOrange,very thick] (0.22,0.225) -- (0.95,0.9);

\draw[BrickRed,very thick] (0.108,0) -- (0.22,0);
\draw[BrickRed,very thick] (23/28,0) -- (0.95,0);

\node[state, fill=Cyan, color=Cyan,minimum size=4pt] at (0.14285714285714285,0.192) {};
\node[state, fill=Cyan, color=Cyan,minimum size=4pt] at (0.8571428571428572,0.8575) {};

\node[state, fill=YellowOrange, color=YellowOrange,minimum size=4pt] at (0.22,0.225) {};
\node[state, fill=YellowOrange, color=YellowOrange,minimum size=4pt] at (0.95,0.9) {};


\end{scope}
 
	\end{tikzpicture}
	\caption{Illustration for the proof of Theorem~\ref{thm:fptas}. Small ticks on the abscissa indicate exponential sampling points of $\cost(\mu_{\theta_2})$ (black). The induced cost of the samples $q^-_1$ and $q^+_1$ are displayed by blue circles. The samples fulfill the sampling property (red). Orange circles indicate the induced cost of $\smash{\mu_{\theta_2}^a}$ and $\smash{\mu_{\theta_2}^b}$ used by the signals of $\phi^*$. Note that $\smash{\cost(q^-_1) \leq \cost(\mu_{\theta_2}^a)}$ and $\smash{\cost(q_1^+) \leq \cost( \mu_{\theta_2}^b )}$ as seen by the dotted lines.
    \label{fig:FPTAS-sketch}}
\end{figure}

In this section, we show that there is an FPTAS for the optimal signaling when there are only two states.
Before we give the formal proof of this result, let us sketch the main arguments.
With \Cref{lem:piecewise-linear} and \Cref{cor:2stateMonotone}, the cost function $C(\mu)$ is piecewise-linear and monotone, see \Cref{fig:FPTAS-sketch} where the orange line shows the cost function induced by the optimal signaling scheme.
The algorithm computes polynomially many sample points for $C$ with exponentially decreasing step size towards the prior as indicated by the ticks on the abscissa.
The algorithm uses an LP to compute the best signaling scheme restricted to the sampling points and the alternative of revealing no information at all. Using Cramer's rule and Hadamard's theorem, it can be shown that a polynomial number of sample points suffice, implying that the algorithm runs in polynomial time. The approximation ratio of $(1+\epsilon)$ for any $\epsilon>0$ is obtained by proving that for any potential optimal conditional belief, there exists a sampling point within an $\epsilon$-distance (red area) that has smaller cost. The cost function induced by these signals guarantees a $(1+\epsilon)$-approximation and is shown in blue in \Cref{fig:FPTAS-sketch}. This is the main result of this section.

\begin{theorem}
    \label{thm:fptas}
    For a single-commodity network $G$ with unknown demands, affine costs, and two states, there is an \mbox{FPTAS} for optimal signaling.
\end{theorem}

\begin{proof}
Let $\Theta = \lbrace \theta_1, \theta_2 \rbrace$ and $0<d_{\theta_1} < d_{\theta_2} = 1$. For simplicity, we represent any probability distribution $\mu \in \Delta(\Theta)$ by its value $\mu_{\theta_2} \in [0,1]$. We denote $\cost(\mu_{\theta_2})$, i.e., $\cost(1)$ and $\cost(0)$ are the social cost when a (deterministic) demand of 1 or $d_{\theta_1}$ routes through the network, respectively.

Our algorithm samples $\cost$ with exponentially decreasing step size around $\prior_{\theta_2}$. Formally, for a given value of $\varepsilon \in (0,1)$, let $\delta := \varepsilon/3$. The algorithm takes samples $\cost$ at points $q^-_j := \prior_{\theta_2} - \prior_{\theta_2} /(1+\delta)^j$ and $q^+_k := \prior_{\theta_2} + (1-\prior_{\theta_2}) /(1+\delta)^k$ for $j=0,1,...,M^-$ and $k=0,1,...,M^+$, where $M^-$ and $M^+$ are polynomials in the input size and $1/\varepsilon$ as discussed below. The algorithm enumerates all pairs $(q^-_j, q^+_k)$ of sample points. For each pair, it constructs the line through $\cost(q_j^-)$ and $\cost(q^+_k)$ defined as
%
    $\ell_{jk}(\lambda) = (\lambda - q^-_j)/(q^+_k - q^-_j) \cost(q^+_k) + (q^+_k - \lambda)/(q^+_k - q^-_j) \cost(q^-_j)$
%
and determines the pair $(q_j^-,q_k^+)$ for which $\ell_{jk}(\prior_2)$ is minimal. 
Eventually, the algorithm resorts to no-signaling if $\cost(\prior_{\theta_2}) \le \min_{j,k} \, \lbrace \ell_{jk}(\prior_{\theta_2}) \rbrace$. Otherwise, the scheme decomposes $\prior$ into signals with conditional distributions $(1-q^-_j,q^-_j)$ and $(1-q^+_k,q^+_k)$ from the optimal pair. Note that in this decomposition, we send the signal for $q_k^+$ with probability $(\prior_{\theta_2} - q_j^-) / (q_k^+ - q_j^-)$, and $q_j^-$ with probability $(q_k^+ - \prior_{\theta_2}) / (q_k^+ -q_j^-)$. Thus, the expected cost of this scheme is indeed $\ell_{jk}(\prior_{\theta_2})$.

\bigskip
{\noindent \textbf{Approximation Ratio.} $\;$}
We first prove that the approximation ratio of the algorithm is upper bounded by $1+\varepsilon$. By Caratheodory's theorem, there is an optimal scheme $\phi^*$ that requires at most $|\Theta| = 2$ signals $\{\sigma^a,\sigma^b\}$~(see, e.g., \cite{Dughmi19}). If these two are the same signal, they both have the conditional distribution $\smash{\prior_{\theta_2}}$. Then the optimal scheme is also found by our algorithm. If $\phi^*$ uses two different signals, they have conditional distributions $\mu_{\theta_2}^a < \prior_{\theta_2} < \mu_{\theta_2}^b$. Since they constitute a convex decomposition of $\prior$, following our reasoning in the last paragraph, the expected cost is
\begin{align*}
    \cost(\phi^*)\! = \! \ell_{ab}(\prior_{\theta_2}) & \!= \!\frac{\prior_{\theta_2} - \mu_{\theta_2}^a}{\mu_{\theta_2}^b - \mu_{\theta_2}^a}  \cost(\mu_{\theta_2}^b) + \frac{\mu_{\theta_2}^b - \prior_{\theta_2}}{\mu_{\theta_2}^b - \mu_{\theta_2}^a} \cost(\mu_{\theta_2}^a).
\end{align*}

We exponentially sample around $\prior_{\theta_2}$ up to a sufficiently small distance. This will guarantee the following \emph{sampling property}: For each possible $\mu_{\theta_2}^a$ and $\mu_{\theta_2}^b$ there exist sample points $q^-_j$ and $q^+_k$ with $q^-_j = \prior_{\theta_2} - (\prior_{\theta_2} - q^-_j) \ge \prior_{\theta_2} - (\prior_{\theta_2} - \mu_{\theta_2}^a) (1+\delta)$ and $q_k^+ = \prior_{\theta_2} + (q_k^+ - \prior_{\theta_2}) \ge \prior_{\theta_2} + (\mu_{\theta_2}^b - \prior_{\theta_2}) / (1+\delta)$. See Figure~\ref{fig:FPTAS-sketch} for an illustration. Furthermore, since $\cost$ is non-decreasing, we have $\cost(q^-_j) \leq \cost(\mu_{\theta_2}^a)$ and $ \cost(q_k^+) \leq \cost(\mu_{\theta_2}^b)$. Therefore, our algorithm computes a signaling scheme with expected cost
\begin{align*}
    \cost_\text{ALG} \; \le \ell_{jk}(\prior_{\theta_2}) &= \frac{\prior_{\theta_2} - q^-_j}{q^+_k - q^-_j}  \cost(q^+_k) + \frac{q^+_k - \prior_{\theta_2}}{q^+_k - q^-_j} \cost(q^-_j) \\
    &\le \; \frac{\prior_{\theta_2} - q^-_j}{q^+_k - q^-_j}  \cost(\mu_{\theta_2}^b) + \frac{q^+_k - \prior_{\theta_2}}{q^+_k - q^-_j} \cost(\mu_{\theta_2}^a).
\end{align*}

As long as $q^-_j < \prior_{\theta_2} < q_k^+$, the partial derivatives of the rightmost expression for $q_j^-$ and $q_k^+$ are 
\begin{align*}
    &\frac{\partial}{\partial q_j^-} \left(\frac{\prior_{\theta_2} - q^-_j}{q^+_k - q^-_j}  \cost(\mu_{\theta_2}^b) + \frac{q^+_k - \prior_{\theta_2}}{q^+_k - q^-_j} \cost(\mu_{\theta_2}^a) \right) =
    - \frac{(\cost(\mu_{\theta_2}^b) - \cost(\mu_{\theta_2}^a))(q^+_k - \prior_{\theta_2})}{(q^+_k - q^-_j)^2} < 0 \quad \text{ and } \\
    &\frac{\partial}{\partial q_k^+} \left(\frac{\prior_{\theta_2} - q^-_j}{q^+_k - q^-_j}  \cost(\mu_{\theta_2}^b) + \frac{q^+_k - \prior_{\theta_2}}{q^+_k - q^-_j} \cost(\mu_{\theta_2}^a) \right) =
    - \frac{(\cost(\mu_{\theta_2}^b) - \cost(\mu_{\theta_2}^a))(\prior_{\theta_2} - q^-_j)}{(q^+_k - q^-_j)^2} < 0,
\end{align*}
so decreasing $q_j^-$ to $\prior_{\theta_2} - (\prior_{\theta_2} - \mu_{\theta_2}^a)  (1+\delta)$ and $q_k^+$ to $\prior_{\theta_2} + (\mu_{\theta_2}^b - \prior_{\theta_2}) / (1+\delta)$ implies
\begin{align*}
    \cost_\text{ALG} &\le \frac{\prior_{\theta_2} - q^-_j}{q^+_k - q^-_j}  \cost(\mu_{\theta_2}^b) + \frac{q^+_k - \prior_{\theta_2}}{q^+_k - q^-_j} \cost(\mu_{\theta_2}^a) \\
    & \le \tfrac{\prior_{\theta_2} - (\prior_{\theta_2}  - (\prior_{\theta_2} - \mu_{\theta_2}^a)(1+\delta))}{\prior_{\theta_2} + (\mu_{\theta_2}^b - \prior_{\theta_2}) / (1+\delta) - (\prior_{\theta_2}  - (\prior_{\theta_2} - \mu_{\theta_2}^a)(1+\delta))}  \cost(\mu_{\theta_2}^b) + \tfrac{\prior_{\theta_2} + (\mu_{\theta_2}^b - \prior) / (1+\delta) - \prior_{\theta_2}}{\prior_{\theta_2} + (\mu_{\theta_2}^b - \prior_{\theta_2}) / (1+\delta) - (\prior_{\theta_2}  - (\prior_{\theta_2} - \mu_{\theta_2}^a)(1+\delta))} \cost(\mu_{\theta_2}^a) \\
    &= \frac{(\prior_{\theta_2} - \mu_{\theta_2}^a) (1+\delta)^2}{(\mu_{\theta_2}^b - \prior_{\theta_2}) + (\prior_{\theta_2} - \mu_{\theta_2}^a)  (1+\delta)^2}  \cost(\mu_{\theta_2}^b)
 + \frac{(\mu_{\theta_2}^b - \prior_{\theta_2})}{(\mu_{\theta_2}^b - \prior_{\theta_2}) + (\prior_{\theta_2} - \mu_{\theta_2}^a)(1+\delta)^2} \cost(\mu_{\theta_2}^a) \\
    &\le (1+\delta)^2 \Biggl(\frac{\prior_{\theta_2} - \mu_{\theta_2}^a}{(\mu_{\theta_2}^b - \prior_{\theta_2}) + (\prior_{\theta_2} - \mu_{\theta_2}^a)}  \cost(\mu_{\theta_2}^b) + \frac{(\mu_{\theta_2}^b - \prior_{\theta_2})}{(\mu_{\theta_2}^b - \prior_{\theta_2}) + (\prior_{\theta_2} - \mu_{\theta_2}^a)} \cost(\mu_{\theta_2}^a) \Biggr) \\
    &= (1+\delta)^2  \cost(\phi^*) \\
    &< (1+\varepsilon)  \cost(\phi^*)
\end{align*}
and proves the approximation ratio.

\bigskip
{\noindent \textbf{Running Time.} $\;$}
The running time crucially relies on the size of parameters $M^-$ and $M^+$. They need to be large enough to ensure that the sampling property holds.
The optimal cost of any decomposition of $\prior$ is determined by the line between two cost values that reaches the lowest value at $\prior_{\theta_2}$. As such, the optimal cost $C(\phi^*)$ as a function of $\prior_{\theta_2}$ constitutes the lower convex envelope of $\cost$.
Since $\cost(\mu)$ is piecewise linear (\Cref{lem:piecewise-linear}) the concave lower envelope is composed of lines between breakpoints of two linear segments. As such, we can assume that $\mu_{\theta_2}^a$ and $\mu_{\theta_2}^b$ from the optimal decomposition of $\prior_{\theta_2}$ are breakpoints of $\cost$. At a breakpoint, we experience a change between two distinct supports for the induced Wardrop equilibrium. 

For a particular distribution $\mu_{\theta_2}$ with support $\support$, the emerging Wardrop equilibrium can be described by the inequalities~\eqref{eq:balance}--\eqref{eq:wardrop-inequality}. We add the constraint $0 \le \mu_{\theta_2} \le 1$. Substituting $\mu_{\theta_1} = 1-\mu_{\theta_2}$ and $y_e = \left( (1-\mu_{\theta_2}) d_1^2 + \mu_{\theta_2}\right)x_e$, the inequalities become linear in variables $\mu_{\theta_2}$, $y_e$ and $\pi_v$:

   \begin{equation}
       \begin{aligned}
       \label{eq:2stateWardrop}
            \pi_v + a_ey_e + b_e ((1-\mu_2) d_1 + \mu_{\theta_2}) &= \pi_w && \text{ for all } e \in \support\\
            \pi_v + a_ey_e + b_e ((1-\mu_2) d_1 + \mu_{\theta_2}) &\geq \pi_w && \text{ for all } e \in E \setminus \support\\
            \sum_{e \in \delta^+(v)} y_e - \sum_{e \in \delta^-(v)} y_e &= \beta_v && \text{ for all } v \in V\\
            y_e &\geq 0 && \text{ for all } e \in E\\
            \pi_s &= 0\\
            \mu_{\theta_2} &\in [0,1].
        \end{aligned}
    \end{equation}

The polytope \eqref{eq:2stateWardrop} describes all distributions $\mu_{\theta_2}$ and emerging Wardrop equilibria for a given support $\support$. Adding an objective function $\max \mu_{\theta_2}$ or $\min \mu_{\theta_2}$, we obtain LPs to find the extremal distributions for which the emerging Wardrop equilibrium has the given support $\support$. Clearly, all breakpoints of $\cost$ are an optimal solution of (at least) two such LPs (one for each distinct support).

Now consider any breakpoint $\mu_{\theta_2}$ and the support $\support$ of its induced Wardrop equilibrium. $\mu_{\theta_2}$ and the Wardrop equilibrium constitute an optimal solution of this LP. Consider the linearly independent constraints fulfilled with equality by this solution, and let $L$ denote the corresponding coefficient matrix of these constraints. Note that the number of rows and columns of $L$ are at most linear in $|V|$ and $|E|$ (and $|\Theta|= 2$). By Cramer's rule, each variable in the optimal solution is a rational number with precision $\det(L)^{-1}$. We may assume that the absolute value of all coefficients in $L$ is bounded by a value $B \le \max\{|a_e|,|b_e| \mid e \in E\} \cup \{1\}$, i.e., $|L_{ij}| \leq B$ for all entries. Hadamard's theorem of determinants yields $\det(L) \leq B^{\tau} \tau^{\tau/2}$ as a trivial upper bound, where $\tau \in \Theta \left( |V| + |E| \right)$. It follows that the precision of $\mu_{\theta_2}$ is limited, i.e., breakpoints live on an exponentially small grid in $[0,1]$. 
To ensure the sampling property, it is sufficient that $\prior_{\theta_2} - \prior_{\theta_2} / (1+\delta)^{M^-} \ge \mu_a$ and $\prior_{\theta_2} + (1 - \prior_{\theta_2}) / (1+\delta)^{M^+} \le \mu_b$ or, equivalently, 
\[
    M^- > \tfrac{\log \tfrac{\prior_{\theta_2}}{\prior_{\theta_2} - \mu_a^{\phantom{*}}}}{\log (1+\delta)} \qquad \text{ and } \qquad
    M^+ > \tfrac{\log \tfrac{1-\prior_{\theta_2}}{\mu_b^{\phantom{*}} - \prior_{\theta_2}}}{\log (1+\delta)}
\]
Suppose $\prior_{\theta_2}$ is represented in the input as a rational number with denominator $B^*$. Since $\mu_a$ is a rational number with denominator $\det(L_a)$, their difference is a rational number with denominator $B^* \cdot \det(L_a)$. Hence $\tfrac{\prior_{\theta_2}}{\prior_{\theta_2} - \mu_a} \le B^* \det(L_a) \le B^* B^{\tau} \tau^{\tau/2}$. A similar bound holds for $\mu_b$. Therefore, 
\[
    M^+, M^- \in \Theta\left(\tfrac{\log B^* + (|V| + |E|) \cdot \log(|V| + |E| + \max_{e \in E} \{a_e,b_e\})}{\varepsilon}\right)
\] 
is sufficient to ensure the sampling property. The running time is bounded by $M^+ + M^- + 3$ evaluations of $\cost(\mu_{\theta_2})$ (which can be done efficiently, e.g., by computing a Wardrop equilibrium), and computing $(M^- + 1)(M^+ + 1)$ points $\ell_{jk}(\prior_{\theta_2})$ to choose the best one.
\end{proof}

\section{Full Information Revelation}
\label{sec:sepa}

As our main result in this section, we show that for a single-commodity network congestion game an optimal signaling scheme always reveals the true state of nature if and only if the underlying network is a series-parallel graph. To prove this result, we recall from \Cref{lem:piecewise-linear} that the cost of the unique Wardrop equilibrium $C$ is a piecewise linear function on $\Delta(\Theta)$. The result then follows from showing that $C$ is concave in $\Delta(\Theta)$ for all cost functions and probability distributions over demands if and only if the underlying network is series-parallel.

 Formally, a graph $G=(V,E)$ with two designated vertices $s,t\in V$ is a series-parallel graph if either it consists only of a single edge $E=\{\{s,t\}\}$, or it is obtained by a parallel or serial composition of two series-parallel graphs.
 For two series-parallel graphs $G_1=(V_1,E_1)$ and $G_2=(V_2,E_2)$ with designated vertices $s_1,t_1\in V_1$ and $s_2,t_2\in V_2$ the \emph{parallel composition} is the graph $G=(V,E)$ created from the disjoint union of graphs $G_1$ and $G_2$ by merging the vertices $s_1$ and $s_2$ into a new vertex $s$, and merging $t_1$ and $t_2$ into a new vertex $t$.
 The \emph{serial composition} of $G_1$ and $G_2$ is the graph created from the disjoint union of graphs $G_1$ and $G_2$ by merging the vertices $t_1$ and $s_2$, and renaming $s_1$ to $s$ and $t_2$ to $t$.
 We treat series-parallel graphs as directed graphs by directing every edge in the orientation as it appears in any path from $s$ to $t$. This is well-defined since in a series-parallel graph, there is a global order on the vertices such that every path only visits vertices in increasing order. 

 The general idea of the proof of the concavity of $C$ on $\Delta(\Theta)$ is as follows. In \Cref{lem:piecewise-linear}, we have shown that $C$ is affine on $P_A$ for all $A\in \mathcal{A}$, i.e., there are affine functions $C_A:\Delta(\Theta)\rightarrow \mathbb{R}_{\geq 0}$ such that $C_A(\mu)=C(\mu)$ for all $\mu \in P_A$. 
 Let furthermore $x_A^*:\Delta(\Theta)\rightarrow \mathbb{R}^E$ be an affine function such that $x^*(\mu)=x_A^*(\mu)$ for all $\mu\in P_A$.
 Note that for $x$ to be a Wardrop equilibrium flow, $x$ must satisfy a system of equations and inequalities explicitly given by \eqref{eq:wardrop-equation} and \eqref{eq:wardrop-inequality} in \Cref{app:full-info}.
 For $\mu \in P_A$, $x_A^*$ is the unique solution $x$ for the system of equations \eqref{eq:wardrop-equation}.
 It is important to note that while $x^*(\mu)=x_A^*(\mu)$ for all $\mu \in P_A$, the vector $x_A^*(\mu)$ with  $\mu \in \Delta(\Theta)\setminus P_A$ will not be a Wardrop equilibrium or not even a feasible flow at all depending on which of the inequalities in \eqref{eq:wardrop-inequality} is violated.

 In the following, we show that the pointwise minimum $\min _{A\in \mathcal{A}}C_A(\mu)$ of all Wardrop equilibria costs always corresponds to a support that is feasible. More specifically, we show in \Cref{lem:smaller-support-demand} that when for some $\mu \in \Delta(\Theta)$, there is a support $A\in \mathcal{A}$ with $\mu \notin P_A$, then there is another support $A'\in \mathcal{A}$ with either lower cost or the same cost but fewer edges. As a consequence, \Cref{lem:WE-minimal} shows that a Wardrop equilibrium is given by the flows $x^*$ that correspond to the pointwise minimum $\min_{A\in \mathcal{A}}C_A(\mu)$, where ties are broken in favor of smaller supports. Finally, our main result is \Cref{thrm:sepa-full-info}. The proofs of this section are very similar to the ones in~\cite{GriesbachHKK22} and, we defer the proof of the following lemma to Appendix~\ref{app:full-info}.
 
\begin{lemma}
    \label{lem:smaller-support-demand}
    For a single-commodity network congestion game on a series-parallel graph, let $A \in \mathcal{A}$ and $\mu \in \Delta(\Theta) \setminus P_A$. Then, there is another support $A' \in \mathcal{A}$ with $C_{A'}(\mu) < C_{A}(\mu)$ or $C_{A'}(\mu) = C_A(\mu)$ and $|A'| < |A|$.
\end{lemma}

Given \Cref{lem:smaller-support-demand}, it is straightforward to show the following result.

\begin{lemma}\label{lem:WE-minimal}
    We have $C(\mu)=\min_{A\in\mathcal{A}}C_A(\mu)$  for all $\mu\in \Delta(\Theta)$.
\end{lemma}

\begin{proof}
    Let $\mu\in\Delta(\Theta)$ be arbitrary and let $A\in\mathcal{A}$ be such that $C_A(\mu)\leq C_{A'}(\mu)$ for all $A'\in \mathcal{A}$ and $|A|<|A'|$ for all $A'\in \mathcal{A}$ with $C_A(\mu)= C_{A'}(\mu)$. As shown in \Cref{lem:smaller-support-demand}, if $x_A^*$ violates one of the inequalities in \eqref{eq:wardrop-inequality}, then there is a support $A''\in \mathcal{A}$ with $C_{A''}(\mu) < C_{A}(\mu)$ or $C_{A''}(\mu) = C_A(\mu)$ and $|A''| < |A|$. This contradicts the choice of $A$.
\end{proof}

We are now in position to show the main result of this section.

\begin{theorem}
\label{thrm:sepa-full-info}
    For a single-commodity network $G$ with unknown demands and affine costs, full information revelation is always an optimal signaling scheme if and only if $G$ is series-parallel.
\end{theorem}

\begin{proof}
    First, we show that full-information revelation is an optimal signaling scheme if $G$ is series-parallel.
    To do so, recall that a signaling scheme $\phi$ is a convex decomposition of the prior $\prior$ into distributions $\mu_\sigma\in \Delta(\Theta)$, and $C(\phi)$ is a convex combination of $C(\mu_\sigma)$, i.e.,
    \begin{align*}
        C(\phi)=\sum_{\sigma\in\Sigma}\phi_\sigma C(\mu_\sigma).
    \end{align*}
    Since $C(\mu)$ is the minimum of affine functions in $\mu$ it is in particular concave in $\mu$. Hence, the best convex decomposition of the prior $\prior$ occurs when there is exactly one individual signal $\sigma_\theta$ for each $\theta \in \Theta$, such that $\mu_{\theta,\sigma}=\chi_\theta$, where $\chi_{\theta}$ is the indicator vector for state $\theta$. 
    
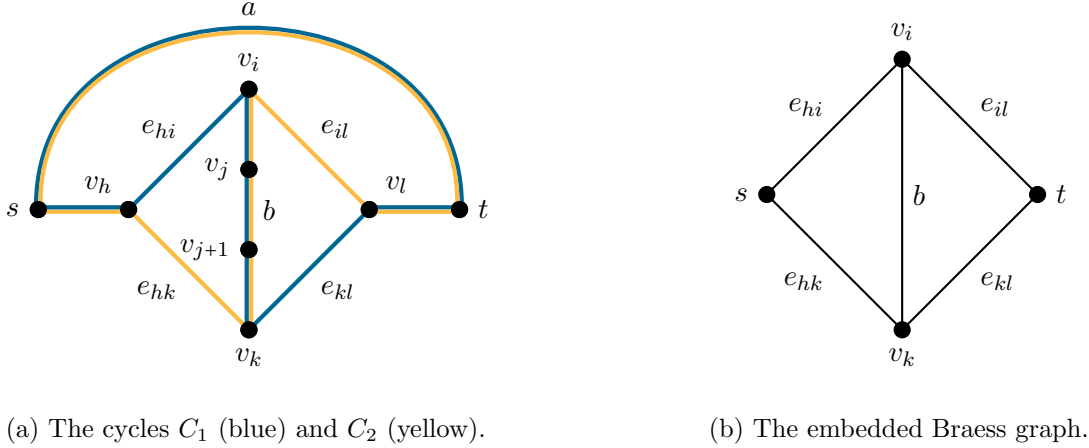
\begin{figure}[tb]
\begin{subfigure}[b]{0.474\textwidth}
\begin{center}
\begin{tikzpicture}[xscale=0.8,yscale=0.8]
    
    \node[state,label=left:{$s$}] (s) at (-1.5,0) [circle] {}; 
    \node[state,label=right:{$t$}] (t) at (5.5,0) [circle] {}; 
    \node[state,label=above:{$v_i$}] (v1) at (2,2) [circle]{}; 
    \node[state,label=below:{$v_k$}] (v2) at (2,-2) [circle]{}; 
    \node[state,label=above left:{$v_h$}] (vh) at (0,0) [circle] {};
    \node[state,label=above right:{$v_l$}] (vl) at (4,0) [circle] {};
    \node[state,label=left:{$v_j$}] (vj) at (2,2/3) [circle]{}; 
    \node[state,label=left:{$v_{j+1}$}] (vj1) at (2,-2/3) [circle]{};

    \path[] (vh) edge[MidnightBlue, ultra thick] node[above left] {\textcolor{black}{$e_{hi}$}} (v1);
    \path[] (vh) edge[Dandelion, ultra thick] node[below left] {\textcolor{black}{$e_{hk}$}} (v2);
    \path[] (v1) edge[Dandelion, ultra thick] node[above right] {\textcolor{black}{$e_{il}$}} (vl);
    \path[] (v2) edge[MidnightBlue, ultra thick] node[below right] {\textcolor{black}{$e_{kl}$}} (vl);
    \path[] (v1.-75) edge[Dandelion, ultra thick] (vj.75);
    \path[] (v1.-105) edge[MidnightBlue, ultra thick] (vj.105);
    \path[] (vj.-75) edge[Dandelion, ultra thick] node[right] {\textcolor{black}{$b$}} (vj1.75);    
    \path[] (vj.-105) edge[MidnightBlue, ultra thick] (vj1.105);
    \path[] (vj1.-105) edge[MidnightBlue, ultra thick] (v2.105);
    \path[] (vj1.-75) edge[Dandelion, ultra thick] (v2.75);
    \path[] (s.-15) edge[Dandelion, ultra thick] (vh.-165);
    \path[] (s.15) edge[MidnightBlue, ultra thick] (vh.165);
    \path[] (vl.15) edge[MidnightBlue, ultra thick] (t.165);	
    \path[] (vl.-15) edge[Dandelion, ultra thick] (t.-165);	
   \draw[MidnightBlue,ultra thick] (s.105) .. controls   (-1.5,4) and (5.5,4) .. node[above,yshift=0mm] {\textcolor{black}{$a$}} (t.75) ;
   \draw[Dandelion, ultra thick] (s.75) .. controls   (-1.5,3.9) and (5.5,3.9) .. (t.105); 

    \node[state] (s) at (-1.5,0) [circle] {}; 
    \node[state] (t) at (5.5,0) [circle] {}; 
    \node[state] (v1) at (2,2) [circle]{}; 
    \node[state] (v2) at (2,-2) [circle]{}; 
    \node[state] (vh) at (0,0) [circle] {};
    \node[state] (vl) at (4,0) [circle] {};
    \node[state] (vj) at (2,2/3) [circle]{}; 
    \node[state] (vj1) at (2,-2/3) [circle]{};

\end{tikzpicture}
\end{center}
\caption{The cycles $C_1$ (blue) and $C_2$ (yellow).}
\label{fig:wheatstone}
\end{subfigure}\hfill
\begin{subfigure}[b]{0.474\textwidth}
\begin{center}
\begin{tikzpicture}[xscale=0.9,yscale=0.9]
    
    \node[state,label=left:{$s$}] (s) at (0,0) [circle] {}; 
    \node[state,label=right:{$t$}] (t) at (4,0) [circle] {}; 
    \node[state,label=above:{$v_i$}] (v1) at (2,2) [circle]{}; 
    \node[state,label=below:{$v_k$}] (v2) at (2,-2) [circle]{}; 
    
    \path[] (s) edge node[above left] {$e_{hi}$} (v1);
    \path[] (s) edge node[below left] {$e_{hk}$} (v2);
    \path[] (v1) edge node[above right] {$e_{il}$} (t);
    \path[] (v2) edge node[below right] {$e_{kl}$} (t);
    \path[] (v1) edge node[right] {$b$} (v2);		
\end{tikzpicture}
\end{center}
\caption{The embedded Braess graph.}
\label{fig:wheatstone-net}
\end{subfigure}
\caption{Illustrations for the proof of Theorem~\ref{thrm:sepa-full-info} }
\end{figure}

    To prove the other direction, it suffices to show that for a non-series-parallel graph, there exist affine cost functions $c:E\rightarrow \mathbb{R}_{\geq0}$ and demands $d_\theta, \theta\in\Theta$ such that full signaling is not optimal.
    For this matter, we first introduce some definitions, which are mostly adopted from Duffin~\cite{Duffin65}.
    We call a graph with two designated vertices $s,t\in V$ a two-terminal graph. Two edges edges $e,e'\in E$ are called \emph{confluent} if there are no two (undirected) simple cycles $C_1$ and $C_2$ both containing $e$ and $e'$ such that the two edges have the same orientation in $C_1$ and different orientations in $C_2$. Further, an edge is \emph{$s$-$t$-confluent} if it is confluent with the (virtually added) edge $(t,s)$. As shown by Duffin, a two-terminal graph $G$ is series-parallel if and only if all edges are $s$-$t$-confluent.
    Let $G=(V,E)$ be a two-terminal graph with source $s\in V$ and sink $t\in V$ such that $G$ is not series-parallel, i.e., $G$ is not $s$-$t$-confluent. Hence, there exist two cycles $C_1$ and $C_2$ which share the (virtually added) edge $a:=(t,s)$ and another edge $b\in E$ such that $a$ is used in the same direction in both cycles but the direction of $b$ changes.
    An illustration is given in \cref{fig:wheatstone}.
    Note that the links $(s,v_h),(v_i,v_j),(v_{j+1},v_k),$ and $(v_l,t)$ consist of a path of an arbitrary amount of edges. In particular, each of these paths can consist of no edge, indicating that the two corresponding  vertices are actually the same.
    On the other hand, the links $(v_h,v_i),(v_h,v_k),(v_i,v_l),$ and $(v_k,v_l)$ consist of paths of at least one edge.
    Therefore, we may choose an arbitrary edge on each of those four links and label them $e_{hi},e_{hk},e_{il}$ and $e_{kl}$, respectively., as shown in \cref{fig:wheatstone}.
    Next, we define the following cost function
    \begin{align*}
        c_e(x)=
        \begin{cases}
            x & \text{if } e\in \{e_{hi},e_{kl}\},\\
            1/2 & \text{if } e\in \{e_{hk},e_{il}\},\\
            1/20 & \text{if } e=b,\\
            \infty & \text{if } e\in E\setminus (E[C_1]\cup E[C_2]),\\
		0 & \text{else}.
        \end{cases}
    \end{align*}
    Contracting all edges of cost $0$ and disregarding all edges of cost $\infty$, we obtain the embedded graph shown in \cref{fig:wheatstone-net}. Together with the above defined cost function, the so constructed graph is just the one from \Cref{exm:not-full-info} for which we have already shown that full information revelation is not an optimal signaling scheme if we set the demands to $d_{\theta_1}=2/5$ and $d_{\theta_2}=1$.
    This shows that for any non-series-parallel graph $G$ we can find cost functions $c:E\rightarrow \mathbb{R}$ and demands $d_\theta, \theta\in \Theta$ such that full signaling is not an optimal signaling scheme.
    This finishes the proof.
\end{proof}

\section{Computing Optimal Schemes}
\label{sec:LPs}
 
We consider the computation of optimal signaling schemes. Towards this end, we first investigate the unique Wardrop equilibria for a fixed set of active edges. We again use the term \emph{support} for a set of active edges. Our approach is generally similar to~\cite{GriesbachHKK22}, but there are notable differences in the analysis to establish the result. Suppose we are given a set of $k$ distinct supports, which we denote by $\support_1,\ldots,\support_k$. Consider the set of signaling schemes $\phi$ with the following properties: $\phi$ sends $k$ signals (where for simplicity we assume $\sigma \in [k] = \{1,\ldots,k\}$), and each signal $\sigma \in [k]$ shall result in a Wardrop equilibrium $x_\sigma$ with support $\support_\sigma$. The main result in this section shows that we can efficiently optimize over this set of signaling schemes. 

\begin{theorem}
    \label{thm:bigLP} 
    Given $k$ distinct support vectors $(\support_{\sigma})_{\sigma \in [k]}$,
    we can compute the best signaling scheme that induces Wardrop equilibria with supports $(\support_{\sigma})_{\sigma \in [k]}$
    in time polynomial in $|\Theta|$, $|E|$ and $k$.  
\end{theorem}

\begin{proof}
    For a single signal $\sigma$ with given support $\support_{\sigma}$, recall from Lemma~\ref{lem:piecewise-linear} the polytope $P_A$ described by~\eqref{eq:balance}-\eqref{eq:wardrop-inequality}. Using the balance vector $(\beta_v)_{v \in V}$ defined in~\eqref{eq:balance}, the system of inequalities reads

   \begin{align}
       \label{eq:wardrop-all}
       \begin{aligned}
            \pi_{v,\sigma} + a_e\left(\sum_{\thetasum} \mu_{\theta,\sigma} d_\theta^2 \right) x_{e,\sigma} + b_e \sum_{\thetasum} \mu_{\theta,\sigma} d_\theta &= \pi_{w,\sigma} && \text{ for all }e \in \support_{\sigma}\\
            \pi_{v,\sigma} + a_e\left(\sum_{\thetasum} \mu_{\theta,\sigma} d_\theta^2 \right) x_{e,\sigma} + b_e \sum_{\thetasum} \mu_{\theta,\sigma} d_\theta &\geq \pi_{w,\sigma}  && \text{ for all } e \in E \setminus \support_{\sigma}\\
            \sum_{e \in \delta^+(v)}  \sum_{\thetasum} \mu_{\theta,\sigma} d_\theta^2 x_{e,\sigma} - \sum_{e \in \delta^-(v)}  \sum_{\thetasum} \mu_{\theta,\sigma} d_\theta^2 x_{e,\sigma} &= \beta_v && \text{ for all } v \in V\\
            x_{e,\sigma} &\geq 0 \quad\phantom{.} && \text{ for all } e \in E\\
            \pi_{s,\sigma} &= 0.\quad&
        \end{aligned}
    \end{align} 

    The conditional distribution $\mu_\sigma$ over states emerges from the signaling probabilities $\phi_{\theta,\sigma}$. It can be captured by the following set of constraints 
	\begin{align}
	\label{eq:condDist}
	\begin{array}{rcll}
		\phi_{\theta,\sigma} &\le& \prior_\theta & \text{ for all } \thetasum, \\
        \phi_{\theta,\sigma} &\ge& 0 &  \text{ for all } \thetasum,\\
		\phi_{\theta,\sigma} &=& \D \phi_\sigma \cdot \mu_{\theta,\sigma} & \text{ for all } \thetasum, \\
        \phi_{\sigma} &=& \D \sum_{\thetasum} \phi_{\theta,\sigma}.
	\end{array}
	\end{align}
	If signal $\sigma$ does not get issued, then $\phi_{\sigma} = 0$, so all $\phi_{\theta,\sigma} = 0$, and the constraints \eqref{eq:wardrop-all} and~\eqref{eq:condDist} are not meaningful. If signal $\sigma$ gets issued, then the polytope of conditional distributions that result in a Wardrop equilibrium on support $\support_{\sigma}$ is described by~\eqref{eq:wardrop-all} and~\eqref{eq:condDist}. We multiply all (in-)equalities of~\eqref{eq:wardrop-all} by $\phi_{\sigma} > 0$. Setting $y_{e,\sigma} = x_{e,\sigma} \cdot \sum_{\thetasum} \phi_{\theta,\sigma} d_\theta^2$ and $\tau_{v,\sigma} = \pi_{v,\sigma} \cdot \phi_{\sigma}$ we obtain the polytope

	\begin{equation}
	\label{eq:signalSupportNetwork}
	\begin{array}{rll}
		\D \tau_{v,\sigma} + a_e y_{e,\sigma} + b_e \sum_{\theta \in \Theta} \phi_{\theta,\sigma} d_\theta & = \tau_{w,\sigma} & \text{ for all } e \in \support_{\sigma} \\
		\D \tau_{v,\sigma} + a_e y_{e,\sigma} + b_e \sum_{\theta \in \Theta} \phi_{\theta,\sigma} d_\theta &  \ge  \tau_{w,\sigma} & \text{ for all } e \in E \setminus \support_{\sigma}\\
		\D \sum_{e \in \delta^+(v)} y_{e,\sigma} - \sum_{e \in \delta^{-}(v)} y_{e,\sigma} & = \beta_v \D \sum_{\thetasum} \phi_{\theta,\sigma} & \text{ for all } v \in V \\
		y_{e,\sigma} &  \ge  0  & \text{ for all } e \in \support_{\sigma} \\
	    \tau_{s,\sigma} &  =  0 \\
		\phi_{\theta,\sigma} &  \le  \prior_\theta  & \text{ for all } \thetasum\\
		\phi_{\theta,\sigma} &  \ge  0  & \text{ for all } \thetasum.
	\end{array}
	\end{equation}

    The scaled balance vector $\beta_v \ \phi_{\sigma}$ on the right-hand side is linear in $\phi_{\theta,\sigma}$, i.e., $\sum_{\thetasum} \phi_{\theta,\sigma} d_\theta^2$ for~$s$, $-\sum_{\thetasum} \phi_{\theta,\sigma} d_\theta^2$ for~$t$, and 0 otherwise. 
    The polytope has a trivial all-zero solution which can be interpreted as the signal not getting issued. Every non-zero solution corresponds to (part of) a signaling scheme $\phi$ that includes a signal $\sigma$ with a resulting Wardrop equilibrium using the given support $\support_{\sigma}$. 
    Finally, all signaling schemes $\phi$ that include such a signal $\sigma$ form a convex polytope in the space of signaling schemes. More formally, we can describe the set of all schemes $\phi$ by combining all polytopes for individual signals $\sigma$ with given support $\support_{\sigma}$ in~\eqref{eq:signalSupportNetwork} and adding the decomposition constraints $\sum_{\sigma \in [k]} \phi_{\theta,\sigma} = \prior_\theta$ for all $\thetasum$.

    We intend to optimize over this polytope of signaling schemes, i.e., we strive to find a scheme with smallest total expected cost $C(\phi)$. The cost can be determined as $
        \cost(\phi) = \sum_{\sigma \in [k]}\phi_\sigma \ \cost(\mu_\sigma) = \sum_{\sigma \in [k]} \phi_{\sigma} \ d \ \pi_{t,\sigma} = \sum_{\sigma \in [k]} \tau_{t,\sigma}.
	$
    As a consequence, finding an optimal signaling scheme for a given set of supports can be formulated as the following linear program. The number of variables and constraints is a polynomial in $|\Theta|$, $|E|$, and $k$. Therefore, the LP can be solved in polynomial time.
	
	\begin{equation}
		\label{eq:schemeSupportNetwork}
		\begin{array}{rll}
			 \multicolumn{2}{l}{\text{Min. } \quad \D \sum_{\sigma \in [k]} \tau_{t,\sigma} } \\
			\text{s.t. } \quad\D \tau_{v,\sigma} + a_e y_{e,\sigma} + \sum_{\thetasum} b_e d_\theta  \phi_{\theta,\sigma} \!\!& = \; \tau_{w,\sigma} & \text{ for all } e \in \support_{\sigma}, \sigma \in [k] \\
			\D \tau_{v,\sigma} + a_e y_{e,\sigma} + \sum_{\theta \in \Theta} b_e d_\theta \phi_{\theta,\sigma}\!\! & \ge  \; \tau_{v,\sigma} & \text{ for all } e \in E \setminus\! \support_{\sigma}, \sigma \in [k]\\
	   	    \D \sum_{e \in \delta^+(v)} y_{e,\sigma} - \sum_{e \in \delta^{-}(v)} y_{e,\sigma} \!\!&= \; \beta_v \D \sum_{\theta \in \thetasum} \phi_{\theta,\sigma} &  \text{ for all } v \in V\\
			y_{e,\sigma} \!\!& \ge \; 0  
            & \text{ for all } e \in \support_{\sigma}, \sigma \in [k] \\
		    \tau_{s,\sigma} \!\!& = \; 0 & \text{ for all } \sigma \in [k] \\
			\D \sum_{\sigma \in [k]} \phi_{\theta,\sigma} \!\!& =  \; \prior_\theta 
            & \text{ for all } \thetasum\\
			\phi_{\theta,\sigma} \!\!& \ge  \; 0  
            & \text{ for all } \thetasum, \sigma \in [k]. 
		\end{array}
	\end{equation}
\end{proof}

This reduces optimizing the signaling scheme to an optimal choice of supports. Suppose for some optimal signaling scheme $\phi^*$ we know (a superset of) all supports $\support_{\sigma}$ used in the Wardrop equilibrium resulting from each signal $\sigma \in \Sigma$ issued in $\phi^*$. 
We inspect the conditions of optimal schemes $\varphi^*$ a bit more closely. Indeed, we need to consider at most $k \le |\Theta|$ signals, and each signal $\sigma$ can be assumed to have a distinct support vector $\support_{\sigma}$.  
 
\begin{proposition}\label{prp:signals}
	There is an optimal signaling scheme $\phi^*$ such that
 at most $|\Theta|$ signals get issued in $\phi^*$, and
  there is no pair of signals $\sigma \neq \sigma'$ that both get issued in $\phi^*$ and $\support_{\sigma} \subseteq \support_{\sigma'}$. In particular, every signal $\sigma$ that gets issued in $\phi^*$ has a distinct support vector $\support_{\sigma}$.
\end{proposition}

\begin{proof}
	The first property follows from Caratheodory's theorem applied in the context of signaling~\cite{Dughmi19}. For the second property, consider an optimal scheme $\phi^*$ resulting from an optimal solution of LP~\eqref{eq:schemeSupportNetwork}. Suppose $\phi^*$ issues two signals $\sigma, \sigma'$ with the subset property for all supports. Then we can ``drop'' signal $\sigma$, i.e., in a new scheme $\phi'$ issues $\sigma'$ whenever we send $\sigma$ or $\sigma'$ in $\phi$. Then $\phi'_{\theta,\sigma} = 0$ and $\phi'_{\theta,\sigma'} = \phi_{\theta,\sigma} + \phi_{\theta,\sigma'}$ for every $\thetasum$. Similarly, $y'_{e,\sigma'} = 0$ and $y'_{e,\sigma} = y_{e,\sigma} + y_{e,\sigma'}$, as well as $\tau_{e,\sigma'} = 0$ and $\tau'_{v,\sigma} = \tau_{v,\sigma} + \tau_{v,\sigma'}$, for every $e \in E$, $v \in V$. Then $(\tau,y,\phi)$ is feasible for LP~\eqref{eq:schemeSupportNetwork} with the same objective function value, i.e, $\phi'$ is also an optimal scheme.
\end{proof}

For a given support $\support$ and two states, we can optimize efficiently over the polytope~\eqref{eq:2stateWardrop} to find the largest and smallest value of $\mu_{\theta_2}$ such that the distribution has a Wardrop equilibrium with support $\support$. Similar to~\cite[Proposition 3]{GriesbachHKK22}, this property can be used to compute \emph{all} supports of Wardrop equilibria for all $\mu \in \Theta(\Delta)$.
\begin{proposition}
	\label{prop:polyTwo}
	The set of all supports of Wardrop equilibria for all $\mu \in \Theta(\Delta)$ in games with two states can be computed in output-polynomial time.
\end{proposition}
When the distributions in $\Delta(\Theta)$ generate at most a polynomial number of different supports in the resulting Wardrop equilibrium, we can compute these supports and, hence, even an optimal signaling scheme in polynomial time. However, there also exist games, in which an exponential number of supports can arise. The instances are nested Braess graphs and emerge as a straightforward adaptation of the constructions in~\cite{KlimmW22,GriesbachHKK22}.
\begin{corollary}
\label{coro:nestedBraess}
	For every number $n \in \NN$, there is a single-commodity game with two states, $O(n)$ vertices, $O(n)$ edges, and $O(n)$ source-target paths, in which $\smash{2^{\Theta(n)}}$ different supports arise in the Wardrop equilibria for all $\mu \in \Delta(\Theta)$.
\end{corollary}

\subsection{Computational Studies}
	\label{ssec:compStudy}

{
    \begin{table}[t]
    \begin{center}
    \begin{tabular}{lcccc}\toprule
	Network & $|V|$ & $|E|$    & $|Z|$ & $d_{\theta_2}$\\\midrule
	Sioux Falls    (SF) & \phantom{2}24 & \phantom{2}76 & 24  & 360,600  \\
	Eastern Massachusetts (EM)  & \phantom{2}74  & 258  & 74 & \phantom{3}65,576 \\
	Berlin-Friedrichshain (BF) & 224 & 523  & 23  & \phantom{3}11,205 \\
	Berlin-Pr.-Berg-Center (BP)  & 352  & 749 & 38 & \phantom{3}16,660 \\
	Berlin-Tiergarten (BT) & 361 & 766  & 26  & \phantom{3}10,755  \\
	Berlin-Mitte-Center (BM) & 398 & 871 & 36 & \phantom{3}11,482 \\\bottomrule
    \end{tabular}
    \caption{Network instances  in the computational studies.}
     \label{tab:cs_networks}   
    \end{center}
    \end{table}
 }
    In the face of Corollary~\ref{coro:nestedBraess} the goal of our study was to investigate i) if instances of our model on realistic networks generate a small number of different supports in the Wardrop equilibrium, and ii) by how much public signaling can improve the total cost in these networks. We considered non-atomic network congestion games with affine costs and uncertain demand on real-world networks for a single commodity and two possible states of nature $\Theta = \lbrace \theta_1, \theta_2 \rbrace$. Table~\ref{tab:cs_networks} shows the six different networks we examined. The network data was obtained from the 
    \citet{CSData22}. The data set includes a model for each network, i.e., it specifies nodes $V$ and links $E$ which correspond to crossings and roads in the real world, respectively. It also defines a partition of the nodes into \emph{zones} $Z$. The size of the networks ranges from smaller ones (SF, EM) to larger ones (BF, BP, BT, BM). The first two are frequently considered in the traffic assignment literature; the latter were used, e.g., by \citet{Jahn05}. 
    
    In addition, the data set provides experimental data on traffic-related properties for each link $e \in E$, such as the capacity $C_e$ and the free-flow travel time $t_e$ (i.e., the time needed to traverse the link in the absence of congestion), and on representative demands between pairs of zones. Originally, the data set is designed for computational studies on the traffic assignment problem with multiple commodities and link cost functions $c^\text{BPR}_e(x)$ as defined in the congestion model of the \citet{BPR64},
	$
    	\smash{c^\text{BPR}_e(x) = t_e  \bigl( 1 + \eta  \bigl( x / C_e \bigr)^\beta \bigr)}.
	$
    Here, $\beta = 4$ and $\eta$ are dimensionless parameters ($\eta = 0.15$ for SF and EM, $\eta = 1$ else). For our model, we defined the coefficients in the cost function $c_e (x) = a_e  x + b_e$ as
	$a_e = \eta t_e  / C_e$ and $b_e = t_e$.  These cost functions correspond to a linear variant of $c^\text{BPR}_e$ (for $\beta = 1$).
    We set the demand in our single-commodity scenario $d_{\theta_2}$ equal to the total demand that is routed through the network for the multi-commodity scenario in the original data (see Table~\ref{tab:cs_networks}). The alternative demand $d_{\theta_1}$ was defined relative to $d_{\theta_2}$, i.e., $d_{\theta_1} = \rho \cdot d_{\theta_2}$ for some $\rho \in [0,1]$. In the following, we show results for $\rho = 0.2$. We performed 40 simulations for each network with varying $(s,t)$-pairs. For each simulation, the $(s,t)$-pair was drawn uniformly at random from the set of zones such that $s \not = t$ and no pair was chosen more than once. Thus, each simulation is given one network and one $(s,t)$-pair. We call such a tuple an \emph{instance}.

    The sets of all supports of $\cost(\mu_{\theta_2})$ over $\mu_{\theta_2} \in [0,1]$ were computed by implementing the approach from Proposition~\ref{prop:polyTwo}, i.e., by recursively computing the support of the emerging Wardrop equilibrium at a mean value for $\mu_{\theta_2}$ (initially $\mu_{\theta_2} = 1/2$), and then solving LP~\eqref{eq:2stateWardrop} twice - once with the objective of maximizing $\mu_{\theta_2}$ and once with the objective of minimizing $\mu_{\theta_2}$. We used the built-in solver of the SciPy package (v1.8.1) \cite{2020SciPy-NMeth}. The flow assignments were computed by an implementation of the conjugate Frank-Wolfe algorithm \cite{FrankWolfe56, Daneva03} in Python (v3.10.6) based on the code of~\citet{CSBasisCode22}. Experiments were performed on an Intel Core i5 based computer at 3.47 GHz with 8 GB RAM operating on Ubuntu 22.04.1 LTS. More information on used libraries and parameters is provided in Appendix~\ref{app:cs}.

        
    \begin{table}[t]
    \begin{center}
     \begin{tabular}{lccccc}\toprule
	  \multirow{2}{*}{Net.}  &  \multicolumn{3}{c}{\rule[2pt]{30pt}{0.2pt}$\;|\mathcal{A}_i|$\;\rule[2pt]{30pt}{0.2pt}}   &   \multicolumn{2}{c}{\rule[2pt]{30pt}{0.2pt}$\;\cost\;$\rule[2pt]{30pt}{0.2pt}}        \\ 
   & AV & SD & MAX & conc.\ $[\%]$ & lin.\ $[\%]$\\
   \midrule
	SF & 4.67 &	2.08 &	\phantom{1}9 &	80 &	10 \\
    EM & 5.15 &	3.14 &	12 &	70 &	\phantom{1}8 \\
    BF & 5.28 &	2.76 &	12 &	68 &	10 \\
    BP & 4.90 &	1.85 &	11 &	88 &	\phantom{1}3 \\
    BT & 5.10 &	2.54 &	11 &	78 &	\phantom{1}8 \\
    BM & 5.15 &	2.38 &	11 &	75 &	\phantom{1}3\\
    \bottomrule
     \end{tabular}
     \caption{Results for the set of all supports $\mathcal{A}_i$ and the concavity and linearity of  $\cost(\mu_{\theta_2})$ for $\mu_{\theta_2} \in [0,1]$ averaged over 40 instances for each network instance.}
    \label{tab:cs_supports}
    \end{center}
    \end{table}

    For each network with instances $i=1, \ldots, 40$, let $\mathcal{A}_i$ be the set of all (distinct) supports of $\cost(\mu_{\theta_2})$. Table~\ref{tab:cs_supports} shows averaged results on the properties of $\mathcal{A}_i$.
    We point out that both the average (AV) and the maximum (MAX) number of used supports turn out to be very small compared to the number of edges in each network, even though the relative difference between $d_{\theta_1}$ and $d_{\theta_2}$ is rather large.
    In fact, these quantities decrease even more for larger values of $\rho$ (for $\rho = 0.8$, the maximum number of used supports ranges from three to five across all instances).
    Moreover, the averaged standard deviation (SD) is small as well. Therefore, these findings imply that computing the optimal signaling scheme in realistic network instances can be done efficiently by solving our approach. The share of instances where $\cost(\mu_{\theta_2})$ is linear is mainly caused by adjacent sources $s$ and targets $t$. 
    The share of concave cost functions reported in Table~\ref{tab:cs_supports} excludes the purely linear cost functions.
    \cref{fig:berlinFH} shows an example for the different supports used in equilibrium for the prior $\prior_{\theta_2}=1/2$.

     \begin{figure}[t]
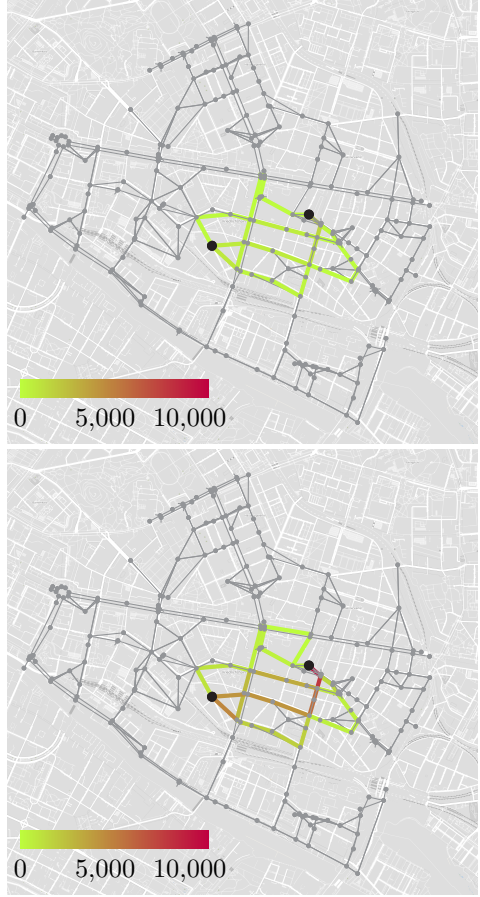

     \centering
    
\tikzstyle{state}=
[circle,draw=Gray,thick,fill=Gray,minimum size=1pt,inner sep=0mm]



     \caption{The map shows the road network of Berlin-Friedrichshain (BF) with the start and destination vertices in black. The used supports are colored from green (only little traffic) to red (gridlocked). If the smaller demand (top) is realized, the optimal signaling scheme always sends the same signal. However, if the large demand (bottom) is realized, it mixes between two different signals. This shifting of the arising equilibria causes a decrease in the overall cost compared to no-signaling or full information revelation.}
     \label{fig:berlinFH}
 \end{figure}

{
    \begin{table}[t]
    \begin{center}
    \begin{tabular}{lccccc}\toprule
	\begin{tabular}{@{}c@{}} Net. \end{tabular} & \begin{tabular}{@{}c@{}} FI is \\ opt.~$[\%]$  \end{tabular} & \begin{tabular}{@{}c@{}} $C$(FI) \\\hline $C$(OPT)  \end{tabular} & \begin{tabular}{@{}c@{}} $C$(NO) \\\hline $C$(OPT)  \end{tabular} & \begin{tabular}{@{}c@{}} $C$(OPT) \\\hline $C$(PSO)  \end{tabular} & \begin{tabular}{@{}c@{}} $C$(WE) \\\hline $C$(PSO)  \end{tabular}\\\midrule
	SF & 100 &	1.0000 &	1.0064 &	1.0135 &	1.0200 \\
EM & 100 &	1.0000 &	1.0052 &	1.0101 &	1.0154 \\
BF & \phantom{1}98 &	1.0000 &	1.0049 &	1.0106 &	1.0156 \\
BP & 100 &	1.0000 &	1.0042 &	1.0091 &	1.0134 \\
BT & 100 &	1.0000 &	1.0051 &	1.0117 &	1.0169 \\
BM & 100 &	1.0000 &	1.0045 &	1.0108 &	1.0154 \\
 \bottomrule
    \end{tabular}
    \caption{Performance of full information revelation (FI), no-signaling (NO), and the optimal signaling scheme (OPT) averaged over 40 instances for each network with $\prior_{\theta_2}=0.5$. The cost of the optimal signaling scheme and the Wardrop equilibrium (WE) are compared to the pointwise social optimum (PSO) defined as $(1-\prior_{\theta_2})  \mathrm{SO}(\mu_{\theta_2} = 0) + \prior_{\theta_2}  \mathrm{SO}(\mu_{\theta_2} = 1)$.}
    \label{tab:cs_costs}
    \end{center}
    \end{table}
}
    For the second part of our study, we analyzed the performance of full information revelation, no-signaling, and the optimal signaling scheme, as shown in Table~\ref{tab:cs_costs}. The results are rounded to four decimal places due to numerical precision. Recall that the cost of the Wardrop equilibrium corresponds to the cost of no-signaling. One can see that in most cases full information revelation is optimal. Moreover, even if it is not optimal, it only produces marginal extra costs compared to the optimal signaling scheme (which are not captured within the numerical precision here). 
    
    
    On another note, the study reveals that using optimal signaling schemes results in slight but consistent improvements over no-signaling.  However, even with optimal information design there remains a notable gap to the average cost of a pointwise social optimal flow.
     As a last remark, Tables~\ref{tab:cs_supports} and \ref{tab:cs_costs} suggest that the optimality criterion of full information revelation goes beyond the resulting Wardrop equilibrium being concave, as used for our characterization in \cref{sec:sepa}, since all networks are not series-parallel.

\appendix
 
\section*{Appendix}

\section{Alternative Interpretation of the Model}
\label{app:alternative}

Each state $\theta$ is drawn with probability $\prior_{\theta}$. Then each agent is active independently with probability $d_{\theta}$. Conditioned on being active, the agents receive a signal $\sigma$, and they update their belief about the demand based on being active and getting signal $\sigma$. This induces a conditional posterior of $\mu_{\sigma}$. An active agent then faces a route selection problem and optimizes by considering the cost of each $s$-$t$-path $P$, which is given by
\begin{align*}
    \sum_{e \in P} c(x_e \mid \mu_\sigma) &= \sum_{e \in P} \sum_{\theta\in\Theta} \mu_{\theta,\sigma}  c_e(d_\theta  x_{e,\sigma})\\
    &= \sum_{e \in P} \sum_{\theta\in\Theta} \mu_{\theta,\sigma} (a_e d_\theta x_{e,\sigma} + b_e)\\
    &= \sum_{e \in P} a_e \left( \sum_{\theta\in\Theta} \mu_{\theta,\sigma} d_\theta\right) x_{e,\sigma} + b_e.
\end{align*}
Here, $x_{e,\sigma}$ corresponds to the amount of flow that routes via edge $e$ when signal $\sigma$ was issued.
Since agents are infinitesimally small and lack information about activity status of other agents in the system, no agent can condition the route choice on the realized demand $\theta$. Given signal $\sigma$, we assume every active agent solves the equilibrium problem using the delays $c(x_e \mid \mu_\sigma)$ and then chooses each path $P$ randomly with probability $x_{P,\sigma}$. Then, for each state $\theta$, the path $P$ gets chosen by a mass of $d_\theta  x_{P,\sigma}$ agents (almost surely). As such, each edge generates a delay of $c_e(d_\theta  x_{e,\sigma})$, exactly as anticipated by all agents. The result is a symmetric Bayesian mixed equilibrium.

Consider a signaling scheme $\phi$, where we define $\phi_{\theta,\sigma}$ as the combined probability that state~$\theta$ is realized and the sender sends signal $\sigma$. For active agents, the scheme $\phi$ induces the posterior 
\[
\mu_{\theta,\sigma} = \frac{\phi_{\theta,\sigma}d_\theta}{\sum_{\theta'\in\Theta} \phi_{\theta',\sigma} d_{\theta'}}
\]
for each signal $\sigma$ with $\sum_{\theta\in\Theta} \phi_{\theta,\sigma} > 0$. The expected cost of any path $P$ for an active agent given signal $\sigma$ is thus given by
\begin{align*}
    &\sum_{e \in P} a_e \left( \sum_{\theta\in\Theta} \mu_{\theta,\sigma} d_\theta\right) x_{e,\sigma} + b_e \\
    &= \sum_{e \in P} a_e \left( \sum_{\theta\in\Theta} \frac{\phi_{\theta,\sigma}d_\theta}{\sum_{\theta'\in\Theta} \phi_{\theta',\sigma} d_{\theta'} } d_\theta\right) x_{e,\sigma} + b_e\\
    &= \frac{1}{\sum_{\theta\in\Theta} \phi_{\theta,\sigma} d_{\theta}}  \sum_{e \in P} a_e \left(\sum_{\theta\in\Theta} \phi_{\theta,\sigma} d_\theta^2\right) x_{e,\sigma} + b_e \left(\sum_{\theta\in\Theta} \phi_{\theta,\sigma} d_{\theta} \right)
\end{align*}
When comparing different routes, the cost of each route involves a uniform scaling factor of $1/\sum_{\theta\in\Theta} \phi_{\theta,\sigma} d_\theta > 0$. As such, the preferences and the emerging equilibrium flow result equivalently from assuming a delay for $P$ of 
\begin{align*}
    \sum_{e \in P} a_e \left(\sum_{\theta\in\Theta} \phi_{\theta,\sigma} d_\theta^2\right) x_{e,\sigma} + b_e \left(\sum_{\theta\in\Theta} \phi_{\theta,\sigma} d_{\theta} \right)
\end{align*}

It is straightforward to verify that the overall cost in the system reads
\begin{align*}
    &\sum_{\theta\in\Theta} \sum_{\sigma\in\Sigma} \phi_{\theta,\sigma}  d_{\theta}  \left(\sum_{e \in P^\sigma} a_e d_\theta x_{e,\sigma} + b_e \right) = \sum_{\sigma\in\Sigma} \sum_{e \in P^\sigma} a_e \left(\sum_{\theta\in\Theta} \phi_{\theta,\sigma} d_\theta^2\right) x_{e,\sigma} + b_e \left( \sum_{\theta\in\Theta} \phi_{\theta,\sigma} d_\theta\right)
\end{align*}
where $P^\sigma$ is a path that carries flow in the equilibrium resulting from signal $\sigma$. 

The model variant is fully equivalent to the variant considered in the paper.

\section{Second Example}
\label{app:prelim}

\begin{figure}[t]
\begin{center}
\begin{subfigure}[b]{0.4\textwidth}
\begin{center}
\begin{tikzpicture}[xscale=2.75,yscale=3]
\tikzstyle{node}=[circle, fill=gray, inner sep=0pt, minimum size=3pt];
\begin{scope}
\node[state] (v1) at (0,0.8) {};
\node[state,label=right:{$t$}] (t1) at (0.5,0.4) {};
\node[state,label=left:{$s$}] (s1) at (-0.5,0.4) {};
\node[state] (w1) at (0,0) {};

\draw[-{Stealth[right]},MidnightBlue,ultra thick] (s1.15) to (v1.-145);
\draw[-{Stealth[left]},Dandelion,ultra thick] (s1.45) to node[above left] {\textcolor{black}{$x$}} (v1.-175);
\draw[-{Stealth},Dandelion,ultra thick] (v1) to node[above right] {\textcolor{black}{$\nicefrac{1}{2}$}} (t1); 
\draw[-{Stealth[right]},LimeGreen,ultra thick] (w1.15) to node[below right] {\textcolor{black}{$x$}} (t1.-145); 
\draw[-{Stealth[left]},MidnightBlue,ultra thick] (w1.45) to (t1.-175);
\draw[-{Stealth},LimeGreen,ultra thick] (s1) to node[below left] {\textcolor{black}{$\nicefrac{1}{2}$}} (w1); 
\draw[-{Stealth},MidnightBlue,ultra thick] (v1) to node[right] {\textcolor{black}{$\nicefrac{1}{20}$}} (w1);

\node[state] (v1) at (0,0.8) {};
\node[state] (t1) at (0.5,0.4) {};
\node[state] (s1) at (-0.5,0.4) {};
\node[state] (w1) at (0,0) {};

\node[state,White] (heightcheat) at (0,-0.5) {};

\end{scope}
\end{tikzpicture}
\caption{}
\end{center}
\end{subfigure}
\hspace{0.5cm}
\begin{subfigure}[b]{0.4\textwidth}
\begin{center}
\begin{tikzpicture}[xscale=3.25,yscale =3.25, shorten >= 0pt,
    shorten <= 0pt,]
\begin{scope}[xshift=15cm,yshift=-1.5cm]

\draw[Red,thick] (0,17/50) -- (2/3,19/25-0.005);
\draw[Red,thick] (2/3,19/25-0.005) -- (1,1-0.005);

\draw[axis,->] (0,0) to (0,1.1) node[above] {$\cost(\mu)$}; 
\draw[axis,->] (0,0) to (1.05,0) node[right] {$\mu_{\theta_2}$};
				
\draw[axis] (1,0.05) -- (1,-0.05) node[below right] {$1$};
\draw[axis] (2/3,0.05) -- (2/3,-0.05) node[below] {$\nicefrac{2}{3}$};
\draw[axis] (2/57,0.05) -- (2/57,-0.05) node[below] {$\nicefrac{2}{57}$};
\draw[axis] (1/2,0.05) -- (1/2,-0.05) node[below] {$\prior_{\theta_2}$}; 

\draw[axis] (0,0) -- (-0.025,0) node[left] {$0$};
\draw[axis] (0.025,1) -- (-0.025,1) node[left] {$1$};
\draw[axis] (0.025,19/25) -- (-0.025,19/25) node[left] {$\nicefrac{19}{25}$};
\draw[axis] (0.025,2/5) -- (-0.025,2/5) node[above left] {$\nicefrac{2}{5}$};
\draw[axis] (0.025,17/50) -- (-0.025,17/50) node[below left] {$\nicefrac{17}{50}$};

\draw[thick] (0,17/50) -- (2/57,2/5);
\draw[thick] (2/57,2/5) -- (2/3,19/25);
\draw[thick] (2/3,19/25) -- (1,1);

\draw[ultra thick, Dandelion] (2/57,-0.3)--(1,-0.3);
\draw[ultra thick, LimeGreen] (2/57,-0.35)--(1,-0.35);
\draw[ultra thick, MidnightBlue] (0,-0.4)--(2/3,-0.4); 

\end{scope}
\end{tikzpicture}
\end{center}
\caption{}
\end{subfigure}
\end{center}
\caption{Optimal signaling in Example~\ref{exm:not-full-info}: 
(a) Instance with three $s$-$t$-paths. They are referred to as the upper (yellow), lower (green), and zig-zag (blue) path. 
(b) Cost of the Wardrop equilibrium as function of the conditional distribution parameterized by $\mu_{\theta_2}$ (black), and convex lower envelope of the function (red).
The colored lines underneath indicate which paths are used in the equilibrium for given parameter $\mu_{\theta_2}$.
}
\label{fig:example-not-monotone}
\label{fig:signaling-scheme}
\end{figure}
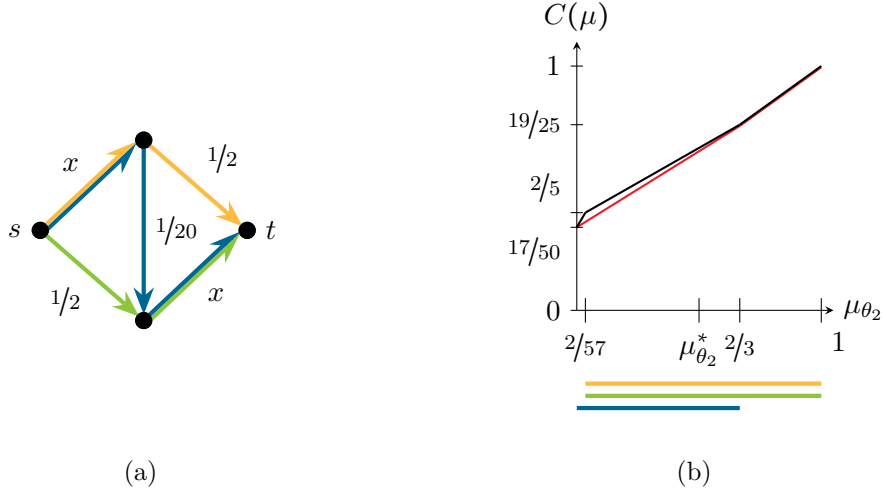

\begin{example}
\label{exm:not-full-info}
We consider a network with four vertices and five edges, see \cref{fig:example-not-monotone}. There are two states $\theta_1$ and $\theta_2$. The demand has a volume of $d_{\theta_1}=2/5$ and $d_{\theta_2}=1$. The network has three $s$-$t$-paths which are referred to as the upper, lower, and zig-zag path, respectively. 

We first analyze the equilibrium for all (conditional) distributions $\mu_{\theta_2}=1-\mu_{\theta_1}\in [0,1]$. For $\mu_{\theta_2}\in [0,2/57]$ the equilibrium uses only the zig-zag path; for $\mu_{\theta_2}\in [2/57,2/3]$ all three paths are used simultaneously; for $\mu_{\theta_2}\in [2/3,1]$ the equilibrium only uses the upper and lower paths.
The resulting total expected cost is shown in black in \cref{fig:example-not-monotone} and given by
\begin{align*}
    C(\mu_{\theta_2})=\
    \begin{cases}
        \tfrac{17}{50}-\tfrac{171}{100}\mu_{\theta_2},
         &\text{for } \mu_{\theta_2}\in \left[\text{\makebox[\widthof{$\tfrac{2}{3}$}]{$0$}},\tfrac{2}{57}\right], \\
        \tfrac{19}{50}-\tfrac{57}{100}\mu_{\theta_2},
        &\text{for } \mu_{\theta_2}\in \left[\tfrac{2}{57},\tfrac{2}{3}\right], \\
        \tfrac{7}{25}-\text{\makebox[\widthof{$\tfrac{171}{100}$}][r]{$\tfrac{18}{25}$}}\mu_{\theta_2},
        &\text{for } \mu_{\theta_2}\in \left[\text{\makebox[\widthof{$\tfrac{2}{57}$}]{$\tfrac{2}{3}$}},\text{\makebox[\widthof{$\tfrac{2}{3}$}]{$1$}}\right]. \\
    \end{cases}
\end{align*}

For instances with two states, there is an elegant intuitive interpretation of the optimal signaling scheme $\phi^*$, which we exploit in more detail in \cref{sec:FPTAS} below. The signaling scheme $\phi^*$ decomposes $\prior$ into a convex combination of conditional distributions, one for each signal. Consider the (black) cost function of the equilibrium in \cref{fig:example-not-monotone}. Suppose we are given a prior $\prior_{\theta_2}=1/2$. An optimal scheme uses two signals with conditional distributions $\mu_l\leq \prior_{\theta_2} \leq \mu_r$. The value of their convex combination is given exactly by the value $\ell_{lr}(\prior_{\theta_2})$, where $\ell_{lr}$ is the affine function that intersects $C(\mu_l)$ and $C(\mu_r)$. Hence, to find $\phi^*$ we want $\mu_l$ and $\mu_r$ such that the value $\ell_{lr}(\prior_{\theta_2})$ is minimal.

In this example, the optimal choices are $\mu_l=0$ and $\mu_r=2/3$. Hence, using two signals $\sigma_1$ and $\sigma_2$, we need to decompose $\prior$ such that the signals have conditional distributions $\mu_{\theta_2,\sigma_1} = \mu_l$ and $\mu_{\theta_2,\sigma_2} = \mu_r$. This implies that $\smash{\phi^*_{\theta_2,\sigma_2} = 1/2}$ and $\smash{\phi^*_{\theta_1,\sigma_1} = \phi^*_{\theta_1,\sigma_2} = 1/4}$, i.e., in $\theta_2$ we always signal $\sigma_2$, in $\theta_1$ we flip a fair coin deciding which signal to send. The induced cost is $C(\phi^*) = \tfrac{131}{200}$. 

In contrast, revealing full information, i.e., deterministically sending a different signal in each state, yields a cost of $\tfrac{134}{200}$, whereas no signal (or always signaling the same signal in all states) yields a cost of $\tfrac{133}{200}$. Hence, $\phi^*$ constructed above is strictly better than both revealing full or no information. Indeed, for all choices of $\prior_{\theta_2}\in[0,2/3]$, the optimal signaling scheme uses signals with $\mu_l$ and $\mu_r$. If $\prior_{\theta_2}\in[2/3,1]$, then it is optimal to send no signal.
Inspecting the induced cost for all possible priors $\prior$, the cost of an optimal signaling scheme as a function of $\prior_{\theta_2}$ corresponds to the convex lower envelope of the function $C(\mu)$, shown in red in \cref{fig:example-not-monotone}.
%
%
%
%
\hfill $\blacksquare$
\end{example}

\section{Proof of \Cref{lem:smaller-support-demand}}
\label{app:full-info}

\begin{proof}
    We prove the statement by induction on $|E|$.

    For $|E|=1$ the graph has only a single edge $e=\{s,t\}$. The only support $A\in \mathcal{A}$ is $A=E$. Clearly, $P_E=\Delta(\Theta)$ and the statement holds trivially since there is no $\mu \in \Delta(\Theta)\setminus P_E$.

    Fix $k\in \mathbb{N}$ and suppose that the statement holds for all series-parallel graphs with up to $k$ edges. Consider a series-parallel graph $G=(V,E)$ with $k+1$ edges. Since $G$ is series-parallel, there is a series of serial and parallel compositions of smaller series-parallel graphs that ends in $G$. In particular, $G$ is constructed either by a final serial composition of two smaller series-parallel graphs $G_1=(V_1,E_1)$ and $G_2=(V_2,E_2)$, or by a final parallel composition of $G_1$ and $G_2$. We proceed to distinguish these two cases.

    \paragraph{First case: Final composition is serial.}
    Let $A\in \mathcal{A}$ with $\mu \in \Delta(\Theta)\setminus P_A$ be given. This implies that either an inequality of type \eqref{eq:wardrop-inequality-1} or an inequality of type \eqref{eq:wardrop-inequality-2} is violated, i.e., either there is an edge $e\in E$ with $(x_A^*)_e<0$ or there is an edge $e=\{v,w\} \in E$ with $c_e((x_A^*)_e\mid\mu)<\pi_w-\pi_v$.

    It is without loss of generality to assume that $e\in E_1$. Let $A_1=A\cap E_1$ and $A_2=A\cap E_2$. For $j\in \{1,2\}$, we denote by $\mathcal{A}_j$ the set of supports for $G_j$ such that the corresponding subgraph is connected. We have $A_j\in \mathcal{A}_j$ for all $j\in \{1,2\}$. 

    The graph $G_1=(V_1,E_1)$ has at most $k$ edges, so we can apply the induction hypothesis on $G_1$ and obtain another support $A_1'\in \mathcal{A}_1$ such that $\smash{C_{A_1'}^1(\mu) < C_{A_1}^1(\mu)}$ where for support $T\in \mathcal{A}_j$ the function $C_T^{j}(\cdot)$ refers to the cost of the solution to the linear system \eqref{eq:wardrop-equation} with support $T$ for $G_j$. Since $G_1$ and $G_2$ were composed in series we have for any support $T\in \mathcal{A}$ that
    \begin{align*}
        C_T(\mu)= C_{T\cap E_1}^{1}(\mu)+C_{T\cap E_2}^{2}(\mu).
    \end{align*}
    So, defining $A'=A_1'\cup A_2\in \mathcal{A}$ we obtain
    \begin{align*}
        C_{A'}(\mu)\!=\!C_{A_1'}^{1}(\mu)+C_{A_2}^{2}(\mu)<C_{A_1}^{1}(\mu)+C_{A_2}^{2}(\mu)\!=\!C_{A}(\mu),
    \end{align*}
    as required.

    \paragraph{Second case: Final compositions is parallel.}
    Let again $A\in \mathcal{A}$ with $\mu \in \Delta(\Theta)\setminus P_A$ be given. Again, this implies that either there is an edge $e\in E$ with $(x_A^*)_e<0$ or there is an edge $e=\{v,w\} \in E$ with $c_e((x_A^*)_e\mid\mu)<\pi_w-\pi_v$. 
    It is without loss of generality to assume that $e\in E_1$.
    For $j\in \{1,2\}$, let $\lambda_j=\sum_{e\in \delta^+(s)\cap E_j}(x_A^*)_e-\sum_{e\in \delta^-(s)\cap E_j}(x_A^*)_e$ be the total flow in $G$ send over the parallel component $G_j$. We have $\lambda_1+\lambda_2=1$.

    Let us first assume that $\lambda_1<0$. This implies $A_1\in \mathcal{A}_1$, since there is a non-zero flow in that component, and, thus, $t$ can be reached from $s$.
    Let $A_1'=A_1\setminus (\delta^+(t)\cap E_1)$. Since $\lambda_1<0$, we have $\lambda_2>1$ and, thus, there is a path from $s$ to $t$ in $A_2$ and in particular $A'=A_1'\cup A_2\in \mathcal{A}$.

    In the following, we write $\lambda_j'$ and $\pi_v'$ for the values of $\lambda_j$ and $\pi_v$ for the new support $A'$. 
    Then we have $\lambda_1'=0$ and, hence, $\lambda_2'=1$. As shown by Klimm and Warode~\cite[Corollary~4]{KlimmW22}, the per-unit cost $\pi_t$ of an equilibrium flow (i.e., a flow that satisfies the linear system \eqref{eq:wardrop-equation} for a fixed support) is strictly increasing in $\lambda$.
    This follows from the fact that the inverse of a Laplace matrix is positive. Hence, $\pi_t'<\pi_t$. Since $C_A(\mu)=\pi_td$, the result follows.

    Next, let us assume that $\lambda_1=0$. In that case, a potential issue is that $A_1$ may not be contained in $\mathcal{A}_1$ since there may not be a path from $s$ to $t$ in $A_1$. As a consequence, we may not be able to apply the induction hypothesis on $G_1$.

    If the edge $e$ violates the inequality $c_e((x_A^*)_e\mid\mu)\geq \pi_w-\pi_v$, we let $A_1'=A_1\cup \{e\}$ and $A'=A_1'\cup A_2$. Now, $G_1$ contains a path from $s$ to $t$ with cost strictly less than $\pi_t$. Since $\pi_t$ is continuous and increasing in the flow, for the equilibrium flow for $A'$, we have $\lambda_1'>0$ and $\lambda_2'<1$. Hence, $\pi_t'<\pi_t$ and the result follows.

    If the edge $e$ violates the inequality $(x_A^*)_e\geq0$, we let $A_1'=A_1\setminus\{e\}$ and $A'=A_1'\cup A_2$. Since there was no flow in $G_1$ anyway, this has no impact on the cost of the flow which is equal to $\pi_t d$, and we have constructed a support $A'$ with $C_{A'}(\mu)=C_A(\mu)$ and $|A'|<|A|$.

    Finally, let us assume that $\lambda_1>0$. Then, $A_1\in \mathcal{A}_1$ and we can apply the induction hypothesis on $G_1$ and obtain a support $A_1'$ for $G_1$ such that either $C_{A_1'}(\mu)<C_{A_1}(\mu)$ or $C_{A_1'}(\mu)=C_{A_1}(\mu)$ and $|A_1'|<|A_1|$. We let $A'=A_1'\cup A_2$.
    In the former case, the result follows from the monotonicity of the flows since we have $\lambda_1'>\lambda_1$ and, hence, $\lambda_2'<\lambda_2$ and, thus, the per unit cost in $G_2$ decreased.
    In the latter case, we have $|A'|<|A|$ and $C_{A'}(\mu)=C_{A}(\mu)$ and we are done.
\end{proof}

\section{Computational Studies}
\label{app:cs}

Additional specifications of relevant software libraries and frameworks: NumPy (v1.23.1) 
Pandas (v1.4.2) \cite{Pandas20}, NetworkX (v2.8.4) \cite{Networkx08}, and OpenMatrix (v0.3.3) \cite{Openmatrix15}. The $(s,t)$-pairs were drawn at random using the \texttt{random} modul for Python, where for each instance, the seed was given by the current system time. For the results in Tables~\ref{tab:cs_supports} and~\ref{tab:cs_costs}, we run the simulation for 272 times to ensure that finally each instance in the data is unique, i.e., there are no duplication of $(s,t)$-pairs for each network. In this way, the variety of instances taken into account by our study increased, which makes our results more reliable. We chose $\rho = 0.2$ in the end as a trade-off between notable changes in the demand on the one hand, and the running time and stability of the simulations on the other hand. Specifications of additional parameters: accuracy of $\cost(\text{WE})$: $1 \times 10^{-4}$, maximum number of iterations for computation of flow: $15000$, maximum seconds allowed for assignment of flow: $6 \times 10^6$, precision of $\mu_{\theta_2}$: $1 \cdot 10^{-8}$, tolerance for projection in Conjugate Frank-Wolfe algorithm: $1 \cdot 10^{-2}$. 

\newpage
 
\bibliographystyle{abbrvnat}
\bibliography{aaai24}

\end{document}